\documentclass[runningheads,letterpaper,10pt]{llncs}

\usepackage[latin1]{inputenc}
\usepackage[margin=1in]{geometry}

\usepackage{amsmath}
\usepackage{amsfonts,amssymb}
\usepackage{array}
\usepackage[caption=false,font=normalsize,labelfont=sf,textfont=sf]{subfig}
\usepackage{textcomp}
\usepackage{stfloats}
\usepackage{url}
\usepackage{graphicx}
\usepackage{cite}

\newcommand\z{\mspace{2mu}}

\newcommand\satvar{(3,3)}

\newcommand\ope[1]{\langle#1\rangle}

\def\mathbi#1{\textbf{\em #1}}

\spnewtheorem{construction}{Construction}{\bfseries}{\itshape}
\spnewtheorem{observation}{Observation}{\bfseries}{\itshape}

\newcommand{\cogsd}{[1{\cdot}2]}
\newcommand{\cogdup}{[2{\cdot}2]}

\newcommand{\figcbf}[1]{\textbf{#1}}
\usepackage{tikz}

\colorlet{path}{brown!90!blue}

\colorlet{mylightblue}{cyan!30}
\colorlet{mylightgray}{lightgray!50}
\colorlet{mylightpurple}{blue!20}

\colorlet{darkred}{red!70!black}
\colorlet{darkblue}{blue!70!black}

\colorlet{intedgecolor}{gray!70} 
\colorlet{outedgecolor}{black}

\newcommand{\bigbigvcl}[3]{
\node[shape=circle,draw=black,thick, fill=white, scale=1.8] (#2) at #1 {};
\node at #1 {\LARGE #3};
}

\newcommand{\bigvcl}[3]{
\node[shape=circle,draw=black,thick, fill=white, scale=1.6] (#2) at #1 {};
\node at #1 {\footnotesize #3};	
}

\newcommand{\bigtvcl}[4]{
\node[shape=circle,draw=black,thick, fill=#4, scale=1.6] (#2) at #1 {};
\node at #1 {\footnotesize #3};	
}

\newcommand{\vcl}[3]{
\node[shape=circle,draw=black,thick, fill=white, scale=0.8] (#2) at #1 {};
\node at #1 {\footnotesize #3};	
}

\newcommand{\tvcl}[4]{
\node[shape=circle,draw=black,thick, fill=#4, scale=0.8] (#2) at #1 {};
\node at #1 {\scriptsize #3};	
}

\newcommand{\ded}[3]{\path [dashed, thick, draw=#3] (#1) edge (#2);}

\newcommand{\ed}[3]{\path [ultra thick, draw=#3] (#1) edge (#2);}

\newcommand{\dtedthin}[3]{\path [very thick, dotted, draw=#3] (#1) edge (#2);}

\newcommand{\sqedthin}[3]{\path [very thick, dotted, draw=#3] (#1) edge (#2);}

\newcommand{\dted}[3]{\path [dotted, ultra thick, draw=#3] (#1) edge (#2);}

\newcommand{\extpath}[3]{\draw [dotted, line width=3pt, path] (#1) -- (#2) node[draw=path,thick,solid,fill=white,font=\Large,midway] {#3};
}

\newcommand{\dded}[3]{\path [dashed, ultra thick, draw=#3] (#1) edge (#2);}

\begin{document}

\title{Closing the complexity gap of the double distance problem%
\thanks{This work was partially supported by CAPES, CNPq and FAPERJ.}
}

\author{
Lu\'is~Cunha\inst{1}
\and 
Thiago~Lopes\inst{1}
\and
U\'everton Souza\inst{1,2}
\and
Leonard Bohnenk\"amper\inst{3}
\and
Mar\'ilia D. V. Braga\inst{3}
\and 
Jens~Stoye\inst{3}
}
\authorrunning{Cunha et al.}
\institute{
Universidade Federal Fluminense, Niter\'oi, Brazil,
\email{lfignacio@ic.uff.br}, \email{thiago\_nascimento@id.uff.br} 
\and
IMPA - Instituto de Matem\'atica Pura e Aplicada, Rio de Janeiro, Brazil, \email{ueverton.souza@impatech.org.br}
\and
Faculty of Technology and Center for Biotechnology (CeBiTec), Bielefeld University, Germany, \email{leonard.bohnenkaemper@uni-bielefeld.de}, \email{mbraga@cebitec.uni-bielefeld.de}, \email{jens.stoye@uni-bielefeld.de}
}

\maketitle

\begin{abstract}
Genome rearrangement has been an active area of research in computational comparative genomics for the last three decades.
While initially mostly an interesting algorithmic endeavor, now the practical application of rearrangement distance methods and more advanced phylogenetic tasks is becoming common practice, given the availability of many completely sequenced genomes.

Several genome rearrangement models have been developed over time, sometimes with surprising computational properties.
A prominent example is the fact that computing the reversal distance of two \emph{signed} permutations is possible in linear time, while for two \emph{unsigned} permutations it is NP-hard.
Therefore one has to always be careful about the precise problem formulation and complexity analysis of rearrangement problems in order not to be fooled.

The double distance is the minimum number of genomic rearrangements between a singular and a duplicated genome that -- in addition to rearrangements -- are separated by a whole genome duplication.
At the same time it allows to assign the genes of the duplicated genome to the two paralogous chromosome copies that existed right after the duplication event.
Computing the double distance is another example of a tricky hardness landscape:
If the distance measure underlying the double distance is the simple breakpoint distance, the problem can be solved in linear time, while with the more elaborate DCJ distance it is NP-hard.
Indeed, there is a whole family of distance measures, parameterized by an even number $k$, between the breakpoint distance ($k=2$) at the one end and the DCJ distance ($k=\infty$) at the other end.
Only little was known about the hardness border that lies somewhere on the way between these two extremes.
Precisely, beneath the two border cases the (linear) problem complexity was known only for $k=4$ and $k=6$.
In this paper we close the gap, giving a full picture of the hardness landscape when computing the double distance.


\end{abstract}

\keywords{
Comparative genomics \and genome rearrangement \and breakpoint distance \and Double Cut and Join (DCJ) distance \and double distance}

\setcounter{footnote}{0}

\section{Introduction}
In comparative genomics,
different measures are used for computing a distance between two given genomes.
Usually a high-level view is adopted, in which each chromosome of a genome is represented by a sequence of oriented genes~\cite{SAN-1992}.
Homologous genes are considered to be equivalent and represented by the same identifier. 
When two genomes over the same set of genes have exactly one copy of each gene, they are said to form a \emph{canonical pair}. 

In the present work we focus on a family of distances that resides between two well known measures.
At one side we have the \emph{breakpoint distance}, which quantifies the distinct adjacencies between the two genomes, an adjacency being the oriented neighborhood between two genes~\cite{tannier2009multichromosomal}.
At the other side we have the \emph{DCJ distance}, which quantifies the minimum number of double cut and join (DCJ) operations required to transform one genome into the other one, each DCJ representing a large-scale genome rearrangement that could be, for example, an inversion, a translocation, a fusion or a fission~\cite{yancopoulos2005efficient}.

Both breakpoint and DCJ distances of canonical genomes can be directly derived from the \emph{breakpoint graph}, a graph that represents the relation between two genomes~\cite{bafna1993breakpointGraph}.
For canonical genomes the breakpoint graph is a simple collection of cycles, all of them with an even number of edges, and paths. We call $i$-cycle a cycle with $i$ edges and $j$-path a path with $j$ edges. Furthermore, we denote by $c_i$ the total number of $i$-cycles and by $p_j$ the total number of $j$-paths in the breakpoint graph.
Assuming that both genomes have~$n_\ast$ genes, the corresponding breakpoint distance is equal to $n_\ast - \left(c_2 + \frac{p_0}{2}\right)$ 
\cite{tannier2009multichromosomal}.
Similarly, when the considered rearrangements are those modeled by the DCJ operation, the DCJ distance is $ n_\ast - \left(c + \frac{p_{\textsc{e}}}{2}\right)$, where $c$ is the total number of cycles and $p_{\textsc{e}}$ is the number of paths with an even number of edges~\cite{bergeron2006unifying}.

Independently of the underlying measure, the distance formulation is a basic unit for several other combinatorial problems related to genome evolution and ancestral reconstruction. The \emph{median} problem, for example, asks for an ancestor genome that minimizes the sum of its
distances to three given genomes. The \emph{double distance} problem, which is the focus of the present study, computes the distance between two genomes that -- in addition to rearrangements -- are separated by a \emph{whole genome duplication}.
Interestingly, both median and double distance can be solved in polynomial time under the breakpoint distance, but become NP-hard under the DCJ distance~\cite{tannier2009multichromosomal}.

In a previous study, focusing on the double distance, some of the authors of the present work started to explore the space between these two extremes by considering the family of \emph{$\sigma_k$ distances}, defined to be $\textup{d}_k = n_\ast - \sigma_k$, where, for $k=2,4,\ldots,\infty$, $\sigma_k = c_2 + c_4 + \cdots + c_k + \frac{p_0 + p_2 + \cdots + p_{k\!-\!2}}{2}$.
Note that $\textup{d}_2$ is equal to the breakpoint distance while $\textup{d}_\infty$ is equal to the DCJ distance. 
Polynomial time algorithms were then provided for the double distance under $\textup{d}_4$ and under $\textup{d}_6$~\cite{braga2024doubleDistance}.

In this paper we prove that, for any finite $k\geq 8$, the double distance under $\textup{d}_k$ is NP-complete, closing the complexity gap of double distance under the family of $\sigma_k$ distances.
Our proof is a polynomial time reduction from a variant of 3-SAT with the restriction that each variable appears at most three times.
This variant is reduced to the $\sigma_8$ disambiguation problem, that is equivalent to the $\sigma_8$ double distance. 
We then explain how this reduction can be generalized to any~$k \geq 8$.

In addition to being an interesting theoretical result, the landscape of the double distance complexity may shed some light on the complexity of the median problem
under the family of $\sigma_k$ distances. 
Here an important observation is that, although $\sigma_2$ median
can be computed in polynomial time \cite{tannier2009multichromosomal}, it gives results that tend to be very close or equivalent to
one of the input genomes \cite{haghighi2012medians}, which might be not so realistic. In contrast, for $k \geq 4$, the $\sigma_k$ median may give results more distributed with respect to the input genomes. However, after more than five years of efforts by several research groups, the complexity of $\sigma_k$ median remains undetermined for any finite $k\geq4$~\cite{silva2023algorithms}.

\section{Background}

A \emph{chromosome} is a linear or circular DNA molecule and a \emph{genome} is a multiset of chromosomes. 
We represent a chromosome by its sequence of genes, where each \emph{gene} is an oriented DNA fragment.
Homologous genes are represented by the same identifier, while distinct genes are represented by distinct identifiers.
Homologous genes occurring in the same genome can also be called \emph{paralogous}.
A gene~$\mathtt{X}$ is represented by the symbol $\mathtt{X}$ itself if it is read in forward orientation or by the symbol $\overleftarrow{\mathtt{X}}$ if it is read in reverse orientation. 
For example, the sequences $[\mathtt{1}\,\overleftarrow{\mathtt{3}}\,\mathtt{2}\,\mathtt{3}]$ and $(4)$  represent, respectively, a linear (flanked by square brackets) and a circular chromosome (flanked by parentheses), the first composed of four genes (two of them are paralogous)
and the second composed of a single gene.
Note that if a sequence $s$ represents a chromosome $S$, then $S$ can be equally represented by the reverse complement of $s$, denoted by~$\overleftarrow{s}$, obtained by reversing the order and the orientation of the genes in~$s$.
Moreover, if $S$ is a circular chromosome, then it can be equally represented by any circular rotation of $s$ or of~$\overleftarrow{s}$.
We denote by~$\mathcal{I}(\mathbb{G})$ the set of distinct gene identifiers that occur in all chromosomes of genome $\mathbb{G}$.

We can also represent a gene $\mathtt{X}$ by referring to its \emph{extremities}~$\mathtt{X}^h$ (head) and~$\mathtt{X}^t$ (tail).
The \emph{adjacencies} in a chromosome are the neighboring extremities of subsequent genes.
The remaining extremities, that are at the ends of linear chromosomes, are \emph{telomeres}.
In linear chromosome $[\mathtt{1}\z\overleftarrow{\mathtt{3}}\z\mathtt{2}]$, the adjacencies are $\{\mathtt{1}^h\mathtt{3}^h, \mathtt{3}^t\mathtt{2}^t\}$ and the telomeres are $\{\mathtt{1}^t,\mathtt{2}^h\}$. 
Note that an adjacency has no orientation, that is, an adjacency between extremities $\mathtt{1}^h$ and $\mathtt{3}^h$ can be equally represented by $\mathtt{1}^h\mathtt{3}^h$ and by $\mathtt{3}^h\mathtt{1}^h$.
An adjacency can occur between extremities of paralogous genes.
In the particular case of a single-gene circular chromosome, e.g.~$(\mathtt{4})$, an adjacency exceptionally occurs between the extremities of the same gene (here $\mathtt{4}^h\mathtt{4}^t$).
We denote by~$\mathcal{A}(\mathbb{G})$ and 
by~$\mathcal{T}(\mathbb{G})$, respectively, the multisets of adjacencies and of telomeres that occur in all chromosomes of genome~$\mathbb{G}$.

Concerning the layout of its chromosomes, a genome can be \emph{circular}, exclusively composed of circular chromosomes, \emph{linear}, exclusively composed of linear chromosomes, or \emph{mixed}, composed of circular and linear chromosomes.

Concerning its gene content, a genome $\mathbb{S}$ is \emph{singular} if each of its genes has no paralogous counterpart in $\mathbb{S}$.
Similarly, a genome $\mathbb{D}$ is \emph{duplicated} if each of its genes has exactly one paralogous counterpart in~$\mathbb{D}$.
In addition to this, if every adjacency occurs twice in a duplicated genome~$\mathbb{D}$, then~$\mathbb{D}$ is a \emph{doubled genome}.
Examples are given in Table~\ref{tab:genome-types}.

\begin{table*}[ht!]
 \caption{\label{tab:genome-types}
   Examples of a singular, a duplicated and two doubled genomes, with their sets of unique gene identifiers and their multisets of adjacencies and of telomeres.
   Note that the doubled genomes $\mathbb{B}_1$ and $\mathbb{B}_2$ have exactly the same multisets of adjacencies and of telomeres.}
 \begin{center}
  \footnotesize
  \renewcommand{\arraystretch}{1.6}
  \begin{tabular}{m{.2\hsize}m{.2\hsize}m{.35\hsize}}
    \hline
    ~~Singular genome\newline
    ~\fbox{\parbox{2.2cm}{\centering each gene\\ is distinct}} & $\mathbb{S}\!=\!\{(\mathtt{1}\z\overleftarrow{\mathtt{3}}\z\mathtt{2})\z(\mathtt{4})[\mathtt{5}\z\overleftarrow{\mathtt{6}}]\}$ & $\begin{cases}\mathcal{F}(\mathbb{S})\!=\!\{\mathtt{1},\mathtt{2},\mathtt{3},\mathtt{4},\mathtt{5},\mathtt{6}\}\\\mathcal{A}(\mathbb{S})\!=\!\{\mathtt{1}^h\mathtt{3}^h, \mathtt{3}^t\mathtt{2}^t, \mathtt{2}^h\mathtt{1}^t, \mathtt{4}^h\mathtt{4}^t, \mathtt{5}^h\mathtt{6}^h\}\\\mathcal{T}(\mathbb{S})\!=\!\{\mathtt{5}^t, \mathtt{6}^t\}\end{cases}$\\[4.2ex]
    \hline
    Duplicated genome\newline 
    ~\fbox{\parbox{2.2cm}{\centering each gene has\\ exactly one\\ paralogous \\ counterpart}} & $\mathbb{D}\!=\!\{(\mathtt{1\,2}\,\overleftarrow{\mathtt{3}}\,\mathtt{1})[\mathtt{3}\,\overleftarrow{\mathtt{2}}]\}$ & $\begin{cases}\mathcal{F}(\mathbb{D})\!=\!\{\mathtt{1},\mathtt{2},\mathtt{3}\}\\
    \mathcal{A}(\mathbb{D})\!=\!\{\mathtt{1}^h\mathtt{2}^t, \mathtt{2}^h\mathtt{3}^h, \mathtt{3}^t\mathtt{1}^t, \mathtt{1}^h\mathtt{1}^t, \mathtt{3}^h\mathtt{2}^h\}\\\mathcal{T}(\mathbb{D})\!=\!\{\mathtt{3}^t, \mathtt{2}^t\}\end{cases}$\\[8ex]
    \hline
    ~Doubled genomes\newline
    ~\fbox{\parbox{2.2cm}{\centering each adjacency\\or telomere\\occurs twice}} & $\mathbb{B}_1\!=\!\{(\mathtt{1}\,\mathtt{2})\,(\mathtt{1}\,\mathtt{2})\,[\mathtt{3}\,\mathtt{4}]\,[\mathtt{3}\,\mathtt{4}]\}$\newline $\mathbb{B}_2\!=\!\{(\mathtt{1}\,\mathtt{2}\,\mathtt{1}\,\mathtt{2})\,[\mathtt{3}\,\mathtt{4}]\,[\mathtt{3}\,\mathtt{4}]\}$ & $\begin{cases}\mathcal{F}(\mathbb{B}_i)\!=\!
    \{\mathtt{1},\mathtt{2},\mathtt{3},\mathtt{4}\}\\
    \mathcal{A}(\mathbb{B}_i)\!=\!\{\mathtt{1}^h\!\mathtt{2}^t\!, \mathtt{2}^h\!\mathtt{1}^t\!, \mathtt{1}^h\!\mathtt{2}^t\!, \mathtt{2}^h\!\mathtt{1}^t\!, \mathtt{3}^h\!\mathtt{4}^t\!, \mathtt{3}^h\!\mathtt{4}^t\}\\\mathcal{T}(\mathbb{B}_i)\!=\!\{\mathtt{3}^t,\mathtt{4}^h,\mathtt{3}^t, \mathtt{4}^h\}\end{cases}$\\[6ex]
    \hline
  \end{tabular}
  \end{center}
\end{table*}

Two genomes $\mathbb{G}_1$ and $\mathbb{G}_2$ are said to be \emph{cognate} when each gene of $\mathbb{G}_1$ has at least one homologous counterpart in~$\mathbb{G}_2$ and \textit{vice versa}.
When two cognate genomes are singular, they are said to form a \emph{canonical pair}.
When a singular genome is a cognate of a duplicated genome, they are said to form a \emph{$\cogsd$-cognate pair}.
Finally, when two cognate genomes are duplicated, they are said to form a \emph{$\cogdup$-cognate pair}.

\subsection{Comparing canonical genomes (with the breakpoint graph)}

The relation between two canonical genomes $\mathbb{S}_1$ and $\mathbb{S}_2$ can be represented by their \emph{breakpoint graph} $BG(\mathbb{S}_1,\mathbb{S}_2) = (V,E)$, that is a multigraph representing adjacencies and telomeres of $\mathbb{S}_1$ and of $\mathbb{S}_2$~\cite{bafna1993breakpointGraph}.
The vertex set $V$ comprises, for each common gene $\mathtt{X}$, one vertex for the extremity $\mathtt{X}^h$ and one vertex for the extremity $\mathtt{X}^t$.
The edge multiset $E$ represents the adjacencies.
For each adjacency in $\mathbb{S}_1$ there exists one $\mathbb{S}_1$-edge in $E$ linking its two extremities, and similar for $\mathbb{S}_2$.

The degree of each vertex can be 0, 1 or 2, therefore each connected \emph{component} is either a cycle or a path. The \emph{length} of a component is given by its number of edges. Note that each component is \emph{alternating}, that is, it switches between $\mathbb{S}_1$- and $\mathbb{S}_2$-edges. As a consequence, all cycles in breakpoint graph must have even length. Furthermore, an even path has one endpoint in $\mathbb{S}_1$ (\emph{$\mathbb{S}_1$-telomere}) and the other endpoint in $\mathbb{S}_2$ (\emph{$\mathbb{S}_2$-telomere}), while an odd path has either both endpoints in $\mathbb{S}_1$ or both endpoints in $\mathbb{S}_2$. Note that an isolated vertex 
is a telomere in both genomes. Such a vertex is a 0-path with one endpoint in $\mathbb{S}_1$ and the other endpoint in $\mathbb{S}_2$.

A vertex that is not a telomere in $\mathbb{S}_1$ nor in $\mathbb{S}_2$ is said to be \emph{non-telomeric}. In the breakpoint graph a non-telomeric vertex has degree 2.
We call \emph{$i$-cycle} a cycle of length~$i$ and \emph{$j$-path} a path of length~$j$. We also denote by $c_i$ the number of $i$-cycles, by $p_j$ the number of $j$-paths, by $c$ the total number of cycles and by $p_{\textsc{e}}$ the total number of even paths: $c=c_2+c_4+c_6+\ldots + c_\infty$ and $p_{\textsc{e}}=p_0+p_2+p_4+\ldots + p_\infty$.
Since the number of telomeres in each genome is even (2 telomeres per linear chromosome), the total number of even paths in the breakpoint graph must be even.
An example is given in Fig.\;\ref{fig:bg}.

\begin{figure}[ht]
  \begin{center}
    \colorlet{qlcolor}{white}

\begin{minipage}{3.5cm}
  \begin{center}
    
   \scalebox{0.65}{ 
    \begin{tikzpicture}[scale=0.9]
		    
  \bigvcl{(0,0)}{u1}{$\mathtt{1}^h$}
  \bigvcl{(0,1.2)}{v1}{$\mathtt{1}^t$}
  \bigvcl{(1.2,1.2)}{v2}{$\mathtt{2}^h$}
  \bigvcl{(1.2,0)}{u2}{$\mathtt{2}^t$}
					
  \bigvcl{(2.4,0)}{u3}{$\mathtt{3}^h$}
  \bigtvcl{(2.4,1.2)}{v3}{$\mathtt{3}^t$}{mylightblue}
  \bigtvcl{(3.6,1.2)}{v4}{$\mathtt{4}^t$}{mylightpurple}
  \bigtvcl{(3.6,0)}{u4}{$\mathtt{4}^h$}{mylightgray}
   
  \ed{v1}{v2}{cyan}
  \path [ultra thick,black,bend left=50] (v1) edge (v2);
			
  \ed{u1}{u2}{cyan}
  \ed{u2}{v3}{black}
  \path [ultra thick,black,bend right=50] (u1) edge (u3);

  \ed{u3}{u4}{cyan}
            
\end{tikzpicture}
}
\end{center}
\end{minipage}
  \end{center}
  \caption{\label{fig:bg}
    Breakpoint graph of canonical genomes $\mathbb{S}_1=\{\,(\mathtt{1\,2})\,[\mathtt{3}\,\overleftarrow{\mathtt{4}}]\,\}$ and $\mathbb{S}_2=\{\,(\mathtt{1}\,\overleftarrow{\mathtt{3}}\,\mathtt{2})\,[\mathtt{4}]\,\}$. Edge types are distinguished by colors: $\mathbb{S}_1$-edges are drawn in blue and $\mathbb{S}_2$-edges are drawn in black. Similarly, vertex types are distinguished by colors: an $\mathbb{S}_1$-telomere is marked in blue, an $\mathbb{S}_2$-telomere is marked in gray, a telomere in both $\mathbb{S}_1$ and $\mathbb{S}_2$ is marked in purple and non-telomeric vertices are white. This graph has one 2-cycle, one 0-path and one 4-path.
  }
\end{figure}

\subsubsection{Breakpoint distance}

For canonical genomes $\mathbb{S}_1$ and $\mathbb{S}_2$ with $n_\ast$ genes each, the \emph{breakpoint distance}, denoted by $\textup{d}_\textsc{bp}$, is defined as follows~\cite{tannier2009multichromosomal}:
$$\textup{d}_\textsc{bp}(\mathbb{S}_1, \mathbb{S}_2) = n_\ast -\left(|\mathcal{A}(\mathbb{S}_1)\cap\mathcal{A}(\mathbb{S}_2)|+\frac{|\mathcal{T}(\mathbb{S}_1)\cap\mathcal{T}(\mathbb{S}_2)|}{2}\right).$$

For $\mathbb{S}_1=\{(\mathtt{1}\z\mathtt{2})\z[\mathtt{3}\z\overleftarrow{\mathtt{4}}]\}$ and $\mathbb{S}_2=\{(\mathtt{1}\z\overleftarrow{\mathtt{3}}\z\mathtt{2})\z[\mathtt{4}]\}$, we have $n_*=4$. The set of common adjacencies is $\mathcal{A}(\mathbb{S}_1)\cap\mathcal{A}(\mathbb{S}_2)=\{\mathtt{1}^t\mathtt{2}^h\}$ and the set of common telomeres is $\mathcal{T}(\mathbb{S}_1)\cap\mathcal{T}(\mathbb{S}_2)=\{\mathtt{4}^t\}$, giving $\textup{d}_\textsc{bp}(\mathbb{S}_1, \mathbb{S}_2)=2.5$.
Since a common adjacency of $\mathbb{S}_1$ and $\mathbb{S}_2$ corresponds to a 2-cycle and a common telomere corresponds to a 0-path in $BG(\mathbb{S}_1, \mathbb{S}_2)$ (see Fig.\;\ref{fig:bg}), the breakpoint distance can be rewritten as
$$\textup{d}_\textsc{bp}(\mathbb{S}_1, \mathbb{S}_2)=n_\ast-\left(c_2+\frac{p_0}{2}\right).$$

\subsubsection{DCJ distance}

Given a genome, a \emph{double cut and join} (DCJ) is the operation that breaks two of its adjacencies or telomeres\footnote{A broken adjacency has two open ends and a broken telomere has a single one.}
 and rejoins the open extremities in a different way~\cite{yancopoulos2005efficient}. 
For example, consider circular chromosome $S=(\z\mathtt{1}\z\mathtt{2}\z\mathtt{3}\z\mathtt{4}\z)$ and a DCJ that cuts $S$ between genes~$\mathtt{1}$ and~$\mathtt{2}$ and between genes~$\mathtt{3}$ and~$\mathtt{4}$, creating segments $\bullet\mathtt{2}\z\mathtt{3}\bullet$ and $\bullet\mathtt{4}\z\mathtt{1}\bullet$ (where the symbol $\bullet$ represents the open ends).
If we join the first with the third and the second with the fourth open end, we get~$S'=(\z\mathtt{1}\z\overleftarrow{\mathtt{3}}\z\overleftarrow{\mathtt{2}}\z\mathtt{4}\z)$, that is, the described DCJ operation is a segmental inversion transforming~$S$ into~$S'$.
Besides inversions, 
DCJ operations can represent several rearrangements, such as translocations, fissions and fusions.
The \emph{DCJ distance} $\textup{d}_\textsc{dcj}$ is the minimum number of DCJs that transform one genome into the other.
For a pair of canonical genomes $\mathbb{S}_1$ and $\mathbb{S}_2$, the DCJ distance can be easily computed with the help of $BG(\mathbb{S}_1, \mathbb{S}_2)$~\cite{bergeron2006unifying}:
$$\begin{array}{ll}
\textup{d}_\textsc{dcj}(\mathbb{S}_1, \mathbb{S}_2)&\!\!\!=n_*-\left(c+\frac{p_{\textsc{e}}}{2}\right)\\
&\!\!\!=n_*-\left(c_2 + c_4 + \ldots + c_{\infty} + \frac{p_0+p_2+\ldots+p_\infty}{2}\right).
\end{array}$$

If $\mathbb{S}_1=\{(\mathtt{1}\z\overleftarrow{\mathtt{3}}\z\mathtt{2})\z[\mathtt{4}]\}$ and $\mathbb{S}_2=\{(\mathtt{1}\z\mathtt{2})\z[\mathtt{3}\z\overleftarrow{\mathtt{4}}]\}$, then $n_*=4$, $c=1$ and $p_{\textsc{e}}=2$ (see Fig.\;\ref{fig:bg}). Consequently, their DCJ distance is $\textup{d}_\textsc{dcj}(\mathbb{S}_1, \mathbb{S}_2)=2$.

\subsubsection{The family of $\sigma_k$ distances}

Given the breakpoint graph of two  canonical genomes $\mathbb{S}_1$ and $\mathbb{S}_2$, for $k \in \{2,4,6,\ldots,\infty\}$, denote by $\sigma_k$ the cumulative sums $\sigma_k=c_2+c_4+\ldots+c_k+\frac{p_0+p_2+\ldots+p_{k\!-\!2}}{2}$.
The \emph{$\sigma_k$ distance} of $\mathbb{S}_1$ and $\mathbb{S}_2$ is defined as~\cite{braga2024doubleDistance}:
$$\textup{d}_k(\mathbb{S}_1,\mathbb{S}_2) = n_\ast - \sigma_k.$$

It is easy to see that the $\sigma_2$ distance equals the breakpoint distance and that the $\sigma_\infty$ distance equals the DCJ distance, and that the distance decreases monotonously between these two extremes.
Moreover, the $\sigma_k$ distance of two genomes that form a canonical pair can easily be computed in linear time for any even~$k \geq 2$.

Observe that the symbol $\infty$ represents unbounded $k$, but in fact the cumulative sums $\sigma_k$ and $\sigma_\infty$ are bounded by $n_*$. Furthermore, since the maximum length of any cycle/path in the breakpoint graph is $2n_*$, for any $k'> k \geq 2n_*$ it holds that $\sigma_{k'} = \sigma_k = \sigma_\infty$.

\subsection{Double distance: comparing \boldmath$\cogsd$-cognate ge\-nomes}

Let a $\cogsd$-cognate pair be composed of a singular genome $\mathbb{S}$ and a duplicated genome $\mathbb{D}$.
The number of genes in $\mathbb{D}$ is twice the number of genes in $\mathbb{S}$ and we need to somehow equalize the contents of these genomes, that is, obtain a $\cogdup$-cognate pair, before searching for common adjacencies of $\mathbb{S}$ and $\mathbb{D}$ or calculating the DCJ distance between the two.

This equalization can be done by \emph{doubling}~$\mathbb{S}$, with a rearrangement operation mimicking a \emph{whole genome duplication}: 
it simply consists of doubling each adjacency and each telomere of~$\mathbb{S}$.
While each linear chromosome is always doubled into two identical copies, it is not possible to find a unique layout of the circular chromosomes obtained after the doubling: each circular chromosome can be doubled into two identical circular chromosomes, or the two copies are concatenated to each other in a single circular chromosome.
Consequently, the doubling of a genome $\mathbb{S}$ results in a set of doubled genomes denoted by $\mathtt{2}\mathbb{S}$.
Note that $|\mathtt{2}\mathbb{S}|=2^{o}$, where $o$ is the number of circular chromosomes in $\mathbb{S}$.
For example, if $\mathbb{S}=\{(\mathtt{1\,2})\,[\mathtt{3}\,\mathtt{4}]\}$, then $\mathtt{2}\mathbb{S}=\{\mathbb{B}_1,\mathbb{B}_2\}$ with
$\mathbb{B}_1=\{(\mathtt{1\,2})\,(\mathtt{1\,2})\,[\mathtt{3}\,\mathtt{4}]\,[\mathtt{3}\,\mathtt{4}]\}$ and
$\mathbb{B}_2=\{(\mathtt{1\,2\,1\,2})\,[\mathtt{3}\,\mathtt{4}]\,[\mathtt{3}\,\mathtt{4}]\}$ (see Table~\ref{tab:genome-types}).
All genomes in $\mathtt{2}\mathbb{S}$ have exactly the same multisets of adjacencies and of telomeres, therefore we can use a special notation for these multisets\footnote{The symbol $\cup$ means that in the resulting multiset 
the multiplicity of an element is equal to the sum of the multiplicities of that element in the original multisets.}: 
$\mathcal{A}(\mathtt{2}\mathbb{S})=\mathcal{A}(\mathbb{S})\cup\mathcal{A}(\mathbb{S})$ and $\mathcal{T}(\mathtt{2}\mathbb{S})=\mathcal{T}(\mathbb{S})\cup\mathcal{T}(\mathbb{S})$.

Note that $\mathbb{D}$ and each $\mathbb{B} \in \mathtt{2}\mathbb{S}$ form a $\cogdup$-cognate pair.
Each pair of paralogous genes in a duplicated genome can be $\binom{\mathtt{a}}{\mathtt{b}}$-\emph{singularized} by adding the index $\mathtt{a}$ to one of its occurrences and the index $\mathtt{b}$ to the other.
In this way, a duplicated genome can be entirely singularized.
Let $\mathfrak{S}^\mathtt{a}_\mathtt{b}(\mathbb{D})$ be the set of all possible genomes obtained by all ways of $\binom{\mathtt{a}}{\mathtt{b}}$-singularizing the duplicated genome $\mathbb{D}$.
Similarly, we denote by $\mathfrak{S}^\mathtt{a}_\mathtt{b}(\mathtt{2}\mathbb{S})$ the set of all possible genomes obtained by all ways of $\binom{\mathtt{a}}{\mathtt{b}}$-singularizing each doubled genome in the set $\mathtt{2}\mathbb{S}$.

\subsubsection{The family of $\sigma_k$ double distances} 

Given a $\cogsd$-cognate pair composed of genomes $\mathbb{S}$ and $\mathbb{D}$ and any genome $\check{\mathbb{D}}$ in $\mathfrak{S}^\mathtt{a}_\mathtt{b}(\mathbb{D})$, the \emph{$\sigma_k$ double distance} of $\mathbb{S}$ and $\mathbb{D}$ for $k=2,4,\ldots,\infty$ is defined as follows~\cite{braga2024doubleDistance}:
$$\textup{dd}_{k}(\mathbb{S}, \mathbb{D}) = \min_{\mathbb{B}\in\mathfrak{S}^\mathtt{a}_\mathtt{b}(\mathtt{2}\mathbb{S})}\{ \textup{d}_{k}(\mathbb{B},\check{\mathbb{D}})\}.$$ 

\noindent
For computing $\textup{dd}_{k}(\mathbb{S}, \mathbb{D})$ we need to optimize the singularizations of $\mathtt{2}\mathbb{S}$ and $\mathbb{D}$.
Note that one can arbitrarily pick an element $\check{\mathbb{D}} \in \mathfrak{S}^\mathtt{a}_\mathtt{b}(\mathbb{D})$ as all other choices from $\mathfrak{S}^\mathtt{a}_\mathtt{b}(\mathbb{D})$ can be generated by renaming the indices.
One then searches for an element in $\mathfrak{S}^\mathtt{a}_\mathtt{b}(\mathtt{2}\mathbb{S})$ that has minimimum $\sigma_k$ double distance with respect to $\check{\mathbb{D}}$.

\subsubsection{$\sigma_2$ (breakpoint), $\sigma_4$ and $\sigma_6$ double distances}

The solution for the $\sigma_2$ double distance, also called \emph{breakpoint double distance}, of $\mathbb{S}$ and~$\mathbb{D}$ can be found easily with a greedy, linear time algorithm~\cite{tannier2009multichromosomal}: the adjacencies and telomeres of $\mathbb{S}$ that occur once or twice in $\mathbb{D}$ induce 
singularizations $\check{\mathbb{D}}$ of~$\mathbb{D}$ and~$\mathbb{B}$ of at least one element of $\mathtt{2}\mathbb{S}$, such that $\check{\mathbb{D}}$ and $\mathbb{B}$ share all these adjacencies and telomeres, in the multiplicity with which they originally occur in $\mathbb{D}$.
Therefore\footnote{The symbol $\cap$ means that in the resulting multiset the multiplicity of an element is equal to the minimum multiplicity of that element in the original multisets.},
$$\textup{dd}_2(\mathbb{S},\mathbb{D})=2n_*-|\mathcal{A}(\mathtt{2}\mathbb{S}) \cap \mathcal{A}(\mathbb{D})|-\frac{|\mathcal{T}(\mathtt{2}\mathbb{S}) \cap \mathcal{T}(\mathbb{D})|}{2},$$
where $n_\ast$ is the number of genes in $\mathbb{S}$.

In a recent study, where some of the authors of the present work participated, linear time algorithms for computing both $\sigma_4$ and $\sigma_6$ double distances of a $\cogsd$-cognate pair of genomes were presented~\cite{braga2024doubleDistance}. These linear solutions can be obtained thanks to the fact that, up to $k=6$, the underlying graph (described in Section~\ref{sec:ambiguousbg} below) can represent arbitrarily large genomes only when its structure is quite particular, similar to a ladder graph, which can be easily handled. When the structure is more complex, it can only represent genomes bound to a limited size. 


\subsubsection{$\sigma_8$ to $\sigma_\infty$ (DCJ) double distances}

For the $\sigma_\infty$ double distance, also called \emph{DCJ double distance}, the solution space cannot be explored efficiently.
In fact, computing the DCJ double distance of circular genomes $\mathbb{S}$ and $\mathbb{D}$ was proven to be an NP-hard problem~\cite{tannier2009multichromosomal}.

In what follows we will present the main results of our present study, which is the NP-completeness of the $\sigma_8$ double distance and, based on this, the NP-completeness of the $\sigma_k$ double distance for any $k \geq 8$. We also generalized the NP-completeness of the DCJ double distance to linear genomes.
This closes the complexity gap of the problem of computing the double distance under the family of $\sigma_k$ distances.

\section{Equivalence of \boldmath$\sigma_k$ double distance and \boldmath$\sigma_k$ disambiguation}

The solution space of the double distance problem can be represented in a compact way with the help of a graph structure introduced by Tannier \textit{et al.}~\cite{tannier2009multichromosomal} and described in the following.

\subsection{Ambiguous breakpoint graph}\label{sec:ambiguousbg}

Given a $\cogsd$-cognate pair composed of a singular genome $\mathbb{S}$ and a duplicated genome $\mathbb{D}$, and a genome $\mathbb{\check{D}} \in \mathfrak{S}^\mathtt{a}_\mathtt{b}(\mathbb{D})$, the \emph{ambiguous breakpoint graph} $ABG(\mathbb{S},\check{\mathbb{D}}) = (V,E)$ is a multigraph representing the adjacencies of $\check{\mathbb{D}}$ and of every element in $\mathfrak{S}^\mathtt{a}_\mathtt{b}(2\mathbb{S})$.
The vertex set $V$ comprises, for each gene identifier $\mathtt{X}$ in $\mathcal{I}(\mathbb{S})$, the two pairs of paralogous vertices $\mathtt{X}^h_\mathtt{a}, \mathtt{X}^h_\mathtt{b}$ and $\mathtt{X}^t_\mathtt{a}, \mathtt{X}^t_\mathtt{b}$.
The edge set $E$ represents the adjacencies.
For each adjacency in $\check{\mathbb{D}}$ there exists one \emph{$\check{\mathbb{D}}$-edge} in $E$ connecting its two extremities. 
The \emph{$\mathbb{S}$-edges} represent all adjacencies occurring in all genomes from $\mathfrak{S}^\mathtt{a}_\mathtt{b}(2\mathbb{S})$:
for each adjacency $\beta\gamma$ of $\mathbb{S}$, we have the pair of paralogous edges $\mathcal{E}(\beta\gamma) = \{ \beta_\mathtt{a}\gamma_\mathtt{a}, \beta_\mathtt{b}\gamma_\mathtt{b} \}$ and the complementary pair of paralogous edges $\Tilde{\mathcal{E}}(\beta\gamma) = \{ \beta_\mathtt{a}\gamma_\mathtt{b}, \beta_\mathtt{b}\gamma_\mathtt{a} \}$.
The \emph{square} of $\beta\gamma$ is then $\mathcal{Q}(\beta\gamma) = \mathcal{E}(\beta\gamma) \cup \Tilde{\mathcal{E}}(\beta\gamma)$.
The $\mathbb{S}$-edges in the ambiguous breakpoint graph are therefore the squares of all adjacencies in $\mathbb{S}$.
If $u$ is a vertex of a square $Q$, we call $\hat{u}$ the \emph{paralogous} vertex of $u$, which is the vertex of $Q$ that is not adjacent to $u$.
Let $a_*$ be the number of squares in $ABG(\mathbb{S}, \check{\mathbb{D}})$.
We then have $a_*=|\mathcal{A}(\mathbb{S})|=n_*-\chi(\mathbb{S})$, where $\chi(\mathbb{S})$ is the number of linear chromosomes in $\mathbb{S}$.
An example of an ambiguous breakpoint graph is depicted in Fig.\;\ref{fig:abg}~(a).

\begin{figure}[ht]
  \begin{center}
    \colorlet{qlcolor}{white}

\begin{minipage}{3cm}
  \begin{center}
    \textbf{(a)}\\

   \scalebox{0.65}{
    \begin{tikzpicture}[scale=0.9]
      \node at (0.6,3.6) {\color{darkred}\footnotesize \bf 1}; 
      \node at (0.6,0.6) {\color{darkred}\footnotesize \bf 2};

      \bigvcl{(0,3)}{q1b1}{$\mathtt{2}^t_\mathtt{b}$}
      \bigvcl{(0,4.2)}{q1t1}{$\mathtt{1}^h_\mathtt{a}$}
      \bigvcl{(1.2,4.2)}{q1t2}{$\mathtt{2}^t_\mathtt{a}$}
      \bigtvcl{(1.2,3)}{q1b2}{$\mathtt{1}^h_\mathtt{b}$}{mylightgray}

      \bigtvcl{(0,0)}{q2b1}{$\mathtt{2}^h_\mathtt{b}$}{mylightgray}
      \bigvcl{(0,1.2)}{q2t1}{$\mathtt{3}^t_\mathtt{b}$}
      \bigvcl{(1.2,1.2)}{q2t2}{$\mathtt{2}^h_\mathtt{a}$}
      \bigvcl{(1.2,0)}{q2b2}{$\mathtt{3}^t_\mathtt{a}$}

      \bigtvcl{(3,0)}{q3b1}{$\mathtt{1}^t_\mathtt{b}$}{mylightblue}
      \bigtvcl{(3,1.2)}{q3t1}{$\mathtt{3}^h_\mathtt{a}$}{mylightblue}
      \bigtvcl{(4.2,1.2)}{q3t2}{$\mathtt{1}^t_\mathtt{a}$}{mylightpurple}
      \bigtvcl{(4.2,0)}{q3b2}{$\mathtt{3}^h_\mathtt{b}$}{mylightpurple}

      \ed{q1b1}{q1b2}{orange}
      \ed{q1t1}{q1t2}{orange}
      \ed{q1b1}{q1t1}{orange}
      \ed{q1b2}{q1t2}{orange}

      \ed{q2b1}{q2b2}{orange}
      \ed{q2t1}{q2t2}{orange}
      \ed{q2b1}{q2t1}{orange}
      \ed{q2b2}{q2t2}{orange}

      \path [ultra thick,black,bend left=50] (q1t1) edge (q1t2);

      \ed{q2t2}{q3t1}{black}
      \ed{q2b2}{q3b1}{black}

      \ed{q1b1}{q2t1}{black}
    \end{tikzpicture}
    }
  \end{center}
\end{minipage}
\hspace{1cm}
\begin{minipage}{3cm}
  \begin{center}
    \textbf{(b)}\\
    
   \scalebox{0.65}{
    \begin{tikzpicture}[scale=0.9]
      \node at (0.6,3.6) {\color{lightgray}\footnotesize 1}; 
      \node at (0.6,0.6) {\color{lightgray}\footnotesize 2};

      \bigvcl{(0,3)}{q1b1}{$\mathtt{2}^t_\mathtt{b}$}
      \bigvcl{(0,4.2)}{q1t1}{$\mathtt{1}^h_\mathtt{a}$}
      \bigvcl{(1.2,4.2)}{q1t2}{$\mathtt{2}^t_\mathtt{a}$}
      \bigtvcl{(1.2,3)}{q1b2}{$\mathtt{1}^h_\mathtt{b}$}{mylightgray}

      \bigtvcl{(0,0)}{q2b1}{$\mathtt{2}^h_\mathtt{b}$}{mylightgray}
      \bigvcl{(0,1.2)}{q2t1}{$\mathtt{3}^t_\mathtt{b}$}
      \bigvcl{(1.2,1.2)}{q2t2}{$\mathtt{2}^h_\mathtt{a}$}
      \bigvcl{(1.2,0)}{q2b2}{$\mathtt{3}^t_\mathtt{a}$}

      \bigtvcl{(3,0)}{q3b1}{$\mathtt{1}^t_\mathtt{b}$}{mylightblue}
      \bigtvcl{(3,1.2)}{q3t1}{$\mathtt{3}^h_\mathtt{a}$}{mylightblue}
      \bigtvcl{(4.2,1.2)}{q3t2}{$\mathtt{1}^t_\mathtt{a}$}{mylightpurple}
      \bigtvcl{(4.2,0)}{q3b2}{$\mathtt{3}^h_\mathtt{b}$}{mylightpurple}
   
      \ed{q1b1}{q1b2}{cyan}
      \ed{q1t1}{q1t2}{cyan}
      \ded{q1b1}{q1t1}{qlcolor}
      \ded{q1b2}{q1t2}{qlcolor}

      \ed{q2b1}{q2b2}{cyan}
      \ed{q2t1}{q2t2}{cyan}
      \ded{q2b1}{q2t1}{qlcolor}
      \ded{q2b2}{q2t2}{qlcolor}

      \path [ultra thick,black,bend left=50] (q1t1) edge (q1t2);

      \ed{q2t2}{q3t1}{black}
      \ed{q2b2}{q3b1}{black}

      \ed{q1b1}{q2t1}{black}
    \end{tikzpicture}
    }
  \end{center}
\end{minipage}

  \end{center}
  \caption{\label{fig:abg}
    \figcbf{(a)} Ambiguous breakpoint graph $ABG(\mathbb{S}, \check{\mathbb{D}})$ for genomes $\mathbb{S}=\{[\mathtt{1\,2\,3}]\}$ and $\check{\mathbb{D}}=\{[\mathtt{1_a\,2_a}\,\overleftarrow{\mathtt{3}}_\mathtt{a}\,\mathtt{1_b}]\,[\overleftarrow{\mathtt{3}}_\mathtt{b}\,\mathtt{2_b}]\}$. Edge types are distinguished by colors: $\check{\mathbb{D}}$-edges are drawn in black and $\mathbb{S}$-edges (squares) are drawn in orange. \figcbf{(b)} Induced breakpoint graph $BG(\tau,\check{\mathbb{D}})$ in which all squares are resolved by the solution $\tau=(\{\mathtt{1}_\mathtt{a}^h\mathtt{2}_\mathtt{a}^t,\mathtt{1}_\mathtt{b}^h\mathtt{2}_\mathtt{b}^t\},\{\mathtt{2}_\mathtt{a}^h\mathtt{3}_\mathtt{b}^t,\mathtt{2}_\mathtt{b}^h\mathtt{3}_\mathtt{a}^t\}\})$, resulting in one 2-cycle, two 0-paths, one 2-path and one 4-path. This is also the breakpoint graph of $\check{\mathbb{D}}$ and $\mathbb{B}=\{[\mathtt{1_a\,2_a\,3_b}],[\mathtt{1_b\,2_b\,3_a}]\} \in \mathfrak{S}^\mathtt{a}_\mathtt{b}(\mathtt{2}\mathbb{S})$. In both (a) and (b), vertex types are distinguished by colors: telomeres in $\mathbb{S}$ are marked in blue, telomeres in $\check{\mathbb{D}}$ are marked in gray, telomeres in both $\mathbb{S}$ and $\check{\mathbb{D}}$ are marked in purple and non-telomeric vertices are white.
  }
\end{figure} 

Each linear chromosome in $\mathbb{S}$ corresponds to four telomeres, called \emph{$\mathbb{S}$-telomeres}, in any element of $\mathtt{2}\mathbb{S}$. These four vertices are not part of any square. In other words, the number of $\mathbb{S}$-telomeres in $ABG(\mathbb{S}, \check{\mathbb{D}})$ is $4\chi(\mathbb{S})$. If $\chi(\mathbb{D})$ is the number of linear chromosomes in $\mathbb{D}$, the number of telomeres in $\check{\mathbb{D}}$, also called \emph{$\check{\mathbb{D}}$-telomeres}, is~$2\chi(\mathbb{D})$. A vertex that is at the same time an $\mathbb{S}$-telomere and a $\check{\mathbb{D}}$-telomere is isolated.

\subsection{Resolving the ambiguous breakpoint graph}

\emph{Resolving} a square $\mathcal{Q}(\cdot) = \mathcal{E}(\cdot) \cup \Tilde{\mathcal{E}}(\cdot)$ corresponds to selecting in the ambiguous breakpoint graph either the edges from $\mathcal{E}(\cdot)$ or the edges from $\Tilde{\mathcal{E}}(\cdot)$, while the complementary pair is deleted.
By resolving all squares, the ambiguous breakpoint graph is resolved.
If we number the squares of $ABG(\mathbb{S},\check{\mathbb{D}})$ from $1$ to $a_\ast$, a \emph{solution} can be represented by a tuple $\tau = (T_1, T_2,\dotsc, T_{a_\ast})$
where each $T_i$ contains the pair of paralogous edges that are chosen in the graph for square $\mathcal{Q}_i$.
The graph \emph{induced} by $\tau$ is a simple breakpoint graph, which we denote by $BG(\tau, \check{\mathbb{D}})$.
An example is shown in Fig.\;\ref{fig:abg}~(b).
The \emph{$k$-score} of~$\tau$, denoted by $\textup{s}_k(\tau)$ is then the sum $\sigma_k$ for $BG(\tau, \check{\mathbb{D}})$. A solution $\tau$ is \emph{$k$-maximum} when, for any solution~$\tau'$ of $ABG(\mathbb{S},\check{\mathbb{D}})$, it holds that $\textup{s}_k(\tau) \geq \textup{s}_k(\tau')$.

Note that each isolated vertex of $ABG(\mathbb{S},\check{\mathbb{D}})$, which is at the same time an $\mathbb{S}$-telomere and a $\check{\mathbb{D}}$-telomere, is a 0-path of any solution of $ABG(\mathbb{S},\check{\mathbb{D}})$, regardless the choices for resolving each square.

\subsubsection{The family of $\sigma_k$ disambiguations}

Given a $\cogsd$-cognate pair composed of genomes $\mathbb{S}$ and $\mathbb{D}$ and any genome $\check{\mathbb{D}}$ in $\mathfrak{S}^\mathtt{a}_\mathtt{b}(\mathbb{D})$, the \emph{$\sigma_k$ disambiguation} of $\mathbb{S}$ and $\mathbb{D}$, denoted by $\textup{ss}_k(\mathbb{S},\mathbb{D})$ for $k=2,4,...,\infty$ is defined as follows~\cite{braga2024doubleDistance}:
$$\textup{ss}_{k}(\mathbb{S}, \mathbb{D}) = \max_{\substack{\tau \text{ is a solution}\\ \text{of } ABG(\mathbb{S},\check{\mathbb{D}})}} \{ \textup{s}_{k}(\tau)\}.$$ 

\noindent
The problem then consists of taking any $\check{\mathbb{D}} \in \mathfrak{S}^\mathtt{a}_\mathtt{b}(\mathbb{D})$ and finding a $k$-maximum solution~$\tau$ of $ABG(\mathbb{S},\check{\mathbb{D}})$. 

\medskip

It was previously shown that the minimization problem of computing the $\sigma_k$ double distance of~$\mathbb{S}$ and $\mathbb{D}$ is equivalent to the maximization problem of obtaining a $\sigma_k$ disambiguation of $\mathbb{S}$ and $\mathbb{D}$~\cite{braga2024doubleDistance}. 
As we have already mentioned, $\sigma_2$, $\sigma_4$ and $\sigma_6$ disambiguations can be solved in linear time, while $\sigma_\infty$ disambiguation is NP-hard~\cite{braga2024doubleDistance,tannier2009multichromosomal}.
We now proceed to show that $\sigma_k$ disambiguation is NP-complete for any finite $k\geq 8$.

\section{NP-completeness of \boldmath$\sigma_k$ disambiguation and of \boldmath$\sigma_k$ double distance}

An \emph{alternating} cycle or path in an ambiguous breakpoint graph alternates between a (square) $\mathbb{S}$-edge and a (connecting) $\check{\mathbb{D}}$-edge. 
A useful graph structure for our proof is called \emph{$p$-flower}. It consists of a circular sequence of $p$ squares, such that, for each square $Q$, one pair of paralogous vertices of $Q$ are matched with a pair of paralogous vertices of its predecessor, while the other pair of paralogous vertices of $Q$ are matched with a pair of paralogous vertices of its successor. An example is given in Fig.\;\ref{fig:p-flower} (a).

\begin{figure}[ht]
  \begin{center}
    \hspace{-0.2cm}
\begin{minipage}{3.2cm}
    \centering
    \textbf{(a)}

    \vspace{-2mm}
    
\scalebox{0.5}{
\begin{tikzpicture}[scale=0.5]

    \vcl{(0.3,2.8)}{q5v1}{}
	\tvcl{(1,4)}{q5v2}{}{yellow}
    \vcl{(-0.2,4.7)}{q5v3}{}
	\tvcl{(-0.9,3.5)}{q5v4}{}{yellow}
    \node at (0.1,3.8) {\large $\mathsf{Q}_5$};

    \vcl{(2.1,5)}{q4v1}{}
	\tvcl{(3.5,5)}{q4v2}{}{yellow}
    \tvcl{(2.1,6.4)}{q4v3}{}{yellow}
	\vcl{(3.5,6.4)}{q4v4}{}
    \node at (2.8,5.7) {\large $\mathsf{Q}_4$};
    
    \tvcl{(5.3,2.8)}{q3v1}{}{yellow}
	\vcl{(6.5,3.5)}{q3v2}{}
    \tvcl{(5.8,4.7)}{q3v3}{}{yellow}
	\vcl{(4.6,4)}{q3v4}{}
    \node at (5.6,3.8) {\large $\mathsf{Q}_3$};

    \tvcl{(1,-0.2)}{q1v1}{}{yellow}
	\vcl{(2,0.8)}{q1v2}{}
    \tvcl{(1,1.8)}{q1v3}{}{yellow}
	\vcl{(0,0.8)}{q1v4}{}
    \node at (1,0.8) {\large $\mathsf{Q}_1$};
    
    \vcl{(4.5,-0.2)}{q2v1}{}
	\tvcl{(5.5,0.8)}{q2v2}{}{yellow}
    \vcl{(4.5,1.8)}{q2v3}{}
	\tvcl{(3.5,0.8)}{q2v4}{}{yellow}
    \node at (4.5,0.8) {\large $\mathsf{Q}_2$};

    \ed{q5v1}{q5v2}{orange}
    \dted{q5v1}{q5v4}{orange}
    \dted{q5v3}{q5v2}{orange}
    \ed{q5v3}{q5v4}{orange}

    \ed{q4v1}{q4v2}{orange}
    \dted{q4v1}{q4v3}{orange}
    \ed{q4v3}{q4v4}{orange}
    \dted{q4v2}{q4v4}{orange}

    \dted{q3v1}{q3v2}{orange}
    \ed{q3v1}{q3v4}{orange}
    \ed{q3v3}{q3v2}{orange}
    \dted{q3v3}{q3v4}{orange}

    \ed{q2v1}{q2v2}{orange}
    \dted{q2v1}{q2v4}{orange}
    \dted{q2v3}{q2v2}{orange}
    \ed{q2v3}{q2v4}{orange}

    \dted{q1v1}{q1v2}{orange}
    \ed{q1v1}{q1v4}{orange}
    \ed{q1v3}{q1v2}{orange}
    \dted{q1v3}{q1v4}{orange}
    
    \ed{q5v2}{q4v1}{black}
    \draw[ultra thick] (q5v4) to[out=120,in=120,distance=3cm] (q4v4);

   
    \ed{q4v2}{q3v4}{black}
    \draw[ultra thick] (q4v3) to[out=60,in=55,distance=3cm] (q3v2);
     
    \ed{q3v1}{q2v3}{black}
    \draw[ultra thick] (q3v3) to[out=0,in=-20,distance=3cm] (q2v1);
    
    \ed{q2v4}{q1v2}{black}
    \draw[ultra thick] (q2v2) to[out=-90,in=-90,distance=3cm] (q1v4);
    
    \ed{q1v3}{q5v1}{black}
    \draw[ultra thick] (q1v1) to[out=-140,in=180,distance=3cm] (q5v3);
\end{tikzpicture}
}
\end{minipage}
   \hspace{-0.6cm}
\begin{minipage}{3.2cm}
    \centering
    \textbf{(b)}

    \vspace{-2mm}
\scalebox{0.5}{    
\begin{tikzpicture}[scale=0.5]

    \vcl{(0.3,2.8)}{q5v1}{}
	\tvcl{(1,4)}{q5v2}{}{white}
    \vcl{(-0.2,4.7)}{q5v3}{}
	\tvcl{(-0.9,3.5)}{q5v4}{}{white}
    \node at (0.1,3.8) {\large $\mathsf{Q}_5$};

    \vcl{(2.1,5)}{q4v1}{}
	\tvcl{(3.5,5)}{q4v2}{}{white}
    \tvcl{(2.1,6.4)}{q4v3}{}{white}
	\vcl{(3.5,6.4)}{q4v4}{}
    \node at (2.8,5.7) {\large $\mathsf{Q}_4$};
    
    \tvcl{(5.3,2.8)}{q3v1}{}{white}
	\vcl{(6.5,3.5)}{q3v2}{}
    \tvcl{(5.8,4.7)}{q3v3}{}{white}
	\vcl{(4.6,4)}{q3v4}{}
    \node at (5.6,3.8) {\large $\mathsf{Q}_3$};

    \tvcl{(1,-0.2)}{q1v1}{}{white}
	\vcl{(2,0.8)}{q1v2}{}
    \tvcl{(1,1.8)}{q1v3}{}{white}
	\tvcl{(0,0.8)}{q1v4}{}{white}
    \node at (1,0.8) {\large $\mathsf{Q}_1$};
    \node (ref1) at (-1.5,-1.5) {};
    \path [ultra thick,black,bend right=50] (q1v4) edge (ref1); 
    \node at (-2.3,-1.2) {\Large $\ope{z}$};
    
    \vcl{(4.5,-0.2)}{q2v1}{}
	\tvcl{(5.5,0.8)}{q2v2}{}{white}
    \vcl{(4.5,1.8)}{q2v3}{}
	\tvcl{(3.5,0.8)}{q2v4}{}{white}
    \node at (4.5,0.8) {\large $\mathsf{Q}_2$};
    \node (ref2) at (7.5,-1.5) {};
    \node (ref1) at (-1.5,-1.5) {};
    \path [ultra thick,black,bend left=50] (q2v2) edge (ref2); 
    \node at (8.3,-1.2) {\Large $\ope{z}$};

    \ed{q5v1}{q5v2}{orange}
    \dted{q5v1}{q5v4}{orange}
    \dted{q5v3}{q5v2}{orange}
    \ed{q5v3}{q5v4}{orange}

    \ed{q4v1}{q4v2}{orange}
    \dted{q4v1}{q4v3}{orange}
    \ed{q4v3}{q4v4}{orange}
    \dted{q4v2}{q4v4}{orange}

    \dted{q3v1}{q3v2}{orange}
    \ed{q3v1}{q3v4}{orange}
    \ed{q3v3}{q3v2}{orange}
    \dted{q3v3}{q3v4}{orange}

    \ed{q2v1}{q2v2}{orange}
    \dted{q2v1}{q2v4}{orange}
    \dted{q2v3}{q2v2}{orange}
    \ed{q2v3}{q2v4}{orange}

    \dted{q1v1}{q1v2}{orange}
    \ed{q1v1}{q1v4}{orange}
    \ed{q1v3}{q1v2}{orange}
    \dted{q1v3}{q1v4}{orange}
    
    \ed{q5v2}{q4v1}{black}
    \draw[ultra thick] (q5v4) to[out=120,in=120,distance=3cm] (q4v4);

   
    \ed{q4v2}{q3v4}{black}
    \draw[ultra thick] (q4v3) to[out=60,in=55,distance=3cm] (q3v2);
     
    \ed{q3v1}{q2v3}{black}
    \draw[ultra thick] (q3v3) to[out=0,in=-20,distance=3cm] (q2v1);
    
    \ed{q2v4}{q1v2}{black}
    
    \ed{q1v3}{q5v1}{black}
    \draw[ultra thick] (q1v1) to[out=-140,in=180,distance=3cm] (q5v3);



\end{tikzpicture}
}
\end{minipage}
\hspace{-0.5cm}
\begin{minipage}{3.2cm}
    \centering
    \textbf{(c)}

    \medskip
    
\scalebox{0.8}{    
\begin{tikzpicture}

\def\xR{0.8}; 
\def\yR{0.3};  

\fill[gray!20] (0,0) ellipse ({\xR} and {\yR}); 
\draw[very thick] (-\xR,0) arc[start angle=180,end angle=0,x radius=0.8,y radius=0.3]; 
\draw[very thick] (-\xR,0) arc[start angle=180,end angle=360,x radius=0.8,y radius=0.3]; 

\path [very thick,black,bend right=50] (-\xR-0.1,0) edge (-\xR-0.5,-1); 

\path [very thick,black,bend left=50] (\xR+0.1,0) edge (\xR+0.5,-1); 

\tvcl{(-\xR-0.1,0)}{}{}{white}

    \tvcl{(\xR+0.1,0)}{}{}{white}

\node at (-1.55,-1) {$\ope{z}$};
 
\node at (1.55,-1) {$\ope{z}$};

\node at (0,0) {\large $\circledast_p$}; 
\end{tikzpicture}
}
\end{minipage}

\vspace{-3mm}
  \end{center}
  \caption{\label{fig:p-flower}
    \figcbf{(a)} An example of a 5-flower.
    In each square the two pairs of paralogous vertices are distinguished by the colors white and yellow.
    \figcbf{(b)} An open 5-flower whose open ends are marked with the symbol $\ope{z}$.
    \figcbf{(c)} Schematic representation of an open $p$-flower whose open ends are marked with the symbol $\ope{z}$.}
\end{figure}

\begin{lemma}\label{lemma:p-flower}
Any solution of a $p$-flower gives either two $2p$-cycles or a single $4p$-cycle.
\end{lemma}
\begin{proof}
Each square is connected twice to each neighbor, via paralogous vertices. For that reason, once we start to walk through any alternating path, we can only close a cycle after visiting each square once. Therefore, the shortest alternating cycle that we can obtain has length $2p$. Indeed, if we select the solid edges of each square, there are two parallel $2p$-cycles, the inner one and the outer one. By switching, i.e., selecting the dotted edges of a square, an even number of times, we obtain two distinct $2p$-cycles, while by switching an odd number of times, we obtain a single $4p$-cycle.
\end{proof}

Now let an \emph{open $p$-flower} be a $p$-flower with exactly one black edge removed (see Fig.\;\ref{fig:p-flower} (b-c)). Note that, as a corollary of Lemma~\ref{lemma:p-flower}, if we exclude from the counting the open edges (marked with $\ope{z}$), any solution of an open $p$-flower gives either an open path of length $4p-1$ or a $2p$-cycle and an open path of length $2p-1$.

\subsection{The problem \boldmath$\satvar$-SAT}

Our proof relies on the classical \emph{satisfiability problem} (SAT), the first to be proven NP-complete. Here we give a brief description of the terminology that we use. 
A boolean \emph{variable} $x$ can be either $\mathtt{true}$ ($\mathtt{T}$) or $\mathtt{false}$ ($\mathtt{F}$). The symbol or \emph{negative literal} $\overline{x}$ represents the negation of $x$, while the \emph{positive literal} $x$ represents the affirmation of $x$. A set of variables $\mathcal{X}$ implies a set of literals $\mathcal{L}(\mathcal{X})=\mathcal{X} \cup \overline{\mathcal{X}}$, where $\overline{\mathcal{X}}=\{\overline{x} \mid x \in \mathcal{X}\}$.
A clause $y=(w_1 \vee w_2 \vee \ldots \vee w_t)$, for $w_i$ being an occurrence of a literal from $\mathcal{L}(\mathcal{X})$, is a logical statement where the symbol $\vee$ represents $\mathtt{OR}$. A clause composed of $t$ literal occurrences is called \emph{$t$-clause}. Denote by $|y|$ the cardinality of clause $y$, that is, $|y|=t$. A \emph{formula} $\mathcal{Y} = [ y_1 \wedge y_2 \wedge \ldots \wedge y_m ]$ is a logical statement where the symbol $\wedge$ represents $\mathtt{AND}$. Denote by $|\mathcal{Y}|$ the cardinality of formula $\mathcal{Y}$, that is, $|\mathcal{Y}|=m$; and by $\|\mathcal{Y}\|$ the number of literal occurrences appearing  in $\mathcal{Y}$, that is, $\|\mathcal{Y}\|=\sum_{y \in \mathcal{Y}} |y|$.
Given a pair $(\mathcal{X},\mathcal{Y})$, where $\mathcal{Y}$ is a formula over the set of variables $\mathcal{X}$, SAT is the problem of determining if there exists an assignment of all variables in $\mathcal{X}$ that satisfies all clauses in $\mathcal{Y}$. 

If the instances of SAT are restricted to have only 2- and 3-clauses, the problem is called 3-SAT. Here we consider a variant of 3-SAT, denoted by $\satvar$-SAT, which additionally has the restriction that each variable appears two or three times and each literal appears once or twice. The variant $\satvar$-SAT is also NP-complete~\cite{tovey1984}. We assume that for each variable that occurs three times, the positive literal appears twice and the negative literal appears once. 
(If this was not the case for a variable $x$, we could replace $x$ by $x'$, where $x' = \overline{x}$.) 
The formula $\mathcal{Y} = [ (x_1 \vee x_2 ) \wedge (x_1 \vee x_3) \wedge (\overline{x_1} \vee \overline{x_2} \vee x_4) \wedge (x_2 \vee x_3) \wedge (\overline{x_3} \vee \overline{x_4}) ]$ over the set of variables $\mathcal{X}=\{x_1,x_2,x_3,x_4\}$ is an instance of $\satvar$-SAT with $|\mathcal{X}|=4$, $|\mathcal{Y}|=5$ and $\|\mathcal{Y}\|=11$. Note that $\mathcal{Y}$ can be satisfied by the assignment $A=\{ (x_1 = \mathtt{T}), (x_2 = \mathtt{T}), (x_3 =\mathtt{F}), (x_4 =\mathtt{T}) \}$, among others.

\subsection{Reducing \boldmath$\satvar$-SAT to the \boldmath$\sigma_8$ disambiguation problem of circular genomes}

We are going to provide a polynomial time reduction from $\satvar$-SAT to the $\sigma_8$ disambiguation problem, but restricted to circular genomes. 
Note that, in this case, in the singular circular genome the number of genes $n_\ast$ equals the number of adjacencies $a_\ast$ (each gene has two extremities and each adjacency is between two extremities).
Furthermore, the ambiguous breakpoint graph has no telomeric vertex and any solution is a simple collection of cycles. 

Our reduction starts with a construction that builds an ambiguous breakpoint graph of circular genomes from an instance of the $\satvar$-SAT problem.

\begin{construction}\label{const:3satAmb}
Given an instance $(\mathcal{X},\mathcal{Y})$ of $\satvar$-SAT, create the graph $G_8(\mathcal{X},\mathcal{Y})$ as follows:
\begin{enumerate}
\item (Variable gadget)
  For each variable $x \in \mathcal{X}$, create a gadget $X$, as shown in Fig.\;\ref{fig:gadgets}~(a1-b), composed of six squares connected to (a1-a2) two open 5-flowers if $x$ occurs three times or to (b) three open 5-flowers if $x$ occurs twice.
\item (Literal occurrence gadget)
  For each literal occurrence $w$ in each clause from $\mathcal{Y}$, create a gadget $W$ composed of two squares connected to an open $5$-flower as shown in Fig.\;\ref{fig:gadgets}~(c).
\item (Clause gadget)
  For each clause $y \in \mathcal{Y}$, create a gadget $Y$ composed of six squares connected as shown in Fig.\;\ref{fig:gadgets}~(d) if $y$ is a 2-clause, and~(e) if $y$ is a 3-clause. (The gadget of a 2-clause includes three 5-flowers.)
\item (Connections between the different gadgets)
  In each clause gadget $Y$ (Fig.\;\ref{fig:gadgets}~(d) or (e)), the edges $\ope{W_1}$ represent the first literal occurrence in clause $y$, that is a literal of variable $x$; these edges are merged\footnote{Merging two pairs of open edges means that each edge of one pair is merged to a distinct edge of the other pair, resulting in two edges.} to the edges $\ope{Y}$ in a literal occurrence gadget $W$ (Fig.\;\ref{fig:gadgets}~(c)). 
  If this connection represents a negative literal, merge the pair of edges $\ope{X}$ in  $W$ to the pair of edges $\ope{\mathtt{F}}$ in the gadget $X$ (Fig.\;\ref{fig:gadgets}~(a2) or~(b)) that represents variable $x$.
  Otherwise, if $x$ occurs three times in $\mathcal{Y}$ (gadget of type $X$-$\mathtt{TTF}$),
  merge them either to the pair of edges $\ope{\mathtt{T}_1}$ or to the pair of edges $\ope{\mathtt{T}_2}$ in $X$. And if $x$ occurs twice in $\mathcal{Y}$ (gadget of type $X$-$\mathtt{TF}$),
  merge them to the pair of edges $\ope{\mathtt{T}}$ in $X$.
  Do the same procedure with the pairs of edges $\ope{W_2}$ and $\ope{W_3}$ of $Y$.
\end{enumerate}
\end{construction}

\begin{figure*}[ht]
    \begin{center}

\begin{minipage}{4cm}
    \begin{center}
   \textbf{(a1) $\mathbi{X}$-$\mathtt{TTF}$ full}

   \smallskip
   
    \scalebox{0.35}{
\begin{tikzpicture}[scale=0.8]

    \node at (2,2) {\huge $\Theta_{\mathtt{F}}$};
    \node at (6,2) {\huge $\Theta_{\mathtt{T}}$};
    
    \vcl{(0,3)}{q1v1}{}
	\vcl{(1,4)}{q1v2}{}
    \vcl{(0,5)}{q1v3}{}
	\vcl{(-1,4)}{q1v4}{}
    \node at (0,4) {\LARGE $\mathsf{Q}^{\z\!X}_1$};

    \vcl{(4,3)}{q2v1}{}
	\vcl{(5,4)}{q2v2}{}
    \vcl{(4,5)}{q2v3}{}
	\vcl{(3,4)}{q2v4}{}
    \node at (4,4) {\LARGE $\mathsf{Q}^{\z\!X}_2$};

    \vcl{(8,3)}{q3v1}{}
	\vcl{(9,4)}{q3v2}{}
    \vcl{(8,5)}{q3v3}{}
	\vcl{(7,4)}{q3v4}{}
    \node at (8,4) {\LARGE $\mathsf{Q}^{\z\!X}_3$};
    
    \vcl{(0,-1)}{q4v1}{}
	\vcl{(1,0)}{q4v2}{}
    \vcl{(0,1)}{q4v3}{}
	\vcl{(-1,0)}{q4v4}{}
    \node at (0,0) {\LARGE $\mathsf{Q}^{\z\!X}_4$};
    
    \vcl{(4,-1)}{q5v1}{}
	\vcl{(5,0)}{q5v2}{}
    \vcl{(4,1)}{q5v3}{}
	\vcl{(3,0)}{q5v4}{}
    \node at (4,0) {\LARGE $\mathsf{Q}^{\z\!X}_5$};
    
    \vcl{(8,-1)}{q6v1}{}
	\vcl{(9,0)}{q6v2}{}
    \vcl{(8,1)}{q6v3}{}
	\vcl{(7,0)}{q6v4}{}
    \node at (8,0) {\LARGE $\mathsf{Q}^{\z\!X}_6$};

    \dted{q1v1}{q1v2}{orange}
    \ed{q1v1}{q1v4}{orange}
    \ed{q1v3}{q1v2}{orange}
    \dted{q1v3}{q1v4}{orange}
    \draw[outedgecolor,ultra thick] (q1v3) -- (0,6) node[anchor=west]{\LARGE $\ope{\mathtt{T}_1}$};
            
    \ed{q2v1}{q2v2}{orange}
    \dted{q2v1}{q2v4}{orange}
    \dted{q2v3}{q2v2}{orange}
    \ed{q2v3}{q2v4}{orange}
    \draw[outedgecolor,ultra thick] (q2v3) -- (4,6) node[anchor=west]{\LARGE $\ope{\mathtt{T}_1}$};

    \dted{q3v1}{q3v2}{orange}
    \ed{q3v1}{q3v4}{orange}
    \ed{q3v3}{q3v2}{orange}
    \dted{q3v3}{q3v4}{orange}
    \draw[outedgecolor,ultra thick] (q3v2) -- (10,4) node[anchor=west]{\LARGE $\ope{\mathtt{F}}$};

    \ed{q4v1}{q4v2}{orange}
    \dted{q4v1}{q4v4}{orange}
    \dted{q4v3}{q4v2}{orange}
    \ed{q4v3}{q4v4}{orange}
    
    \draw[outedgecolor,ultra thick] (q4v1) -- (0,-2) node[anchor=west]{\LARGE $\ope{\mathtt{T}_2}$};

    \dted{q5v1}{q5v2}{orange}
    \ed{q5v1}{q5v4}{orange}
    \ed{q5v3}{q5v2}{orange}
    \dted{q5v3}{q5v4}{orange}
    \draw[outedgecolor,ultra thick] (q5v1) -- (4,-2) node[anchor=west]{\LARGE $\ope{\mathtt{T}_2}$};

    \ed{q6v1}{q6v2}{orange}
    \dted{q6v1}{q6v4}{orange}
    \dted{q6v3}{q6v2}{orange}
    \ed{q6v3}{q6v4}{orange}
    \draw[outedgecolor,ultra thick] (q6v2) -- (10,0) node[anchor=west]{\LARGE $\ope{\mathtt{F}}$};
    
    \ed{q1v2}{q2v4}{intedgecolor}
    \ed{q1v1}{q4v3}{intedgecolor}
    \ed{q2v2}{q3v4}{intedgecolor}
    \ed{q2v1}{q5v3}{intedgecolor}
    \ed{q3v1}{q6v3}{intedgecolor}
    \ed{q4v2}{q5v4}{intedgecolor}
    \ed{q5v2}{q6v4}{intedgecolor}

\def\xR{0.6}; 
\def\yR{1.4};  

\fill[gray!20] (-2.5,2.1) ellipse ({\xR} and {\yR}); 
\draw[very thick] (-\xR-2.5,2.1) arc[start angle=180,end angle=0,x radius=\xR,y radius=\yR]; 
\draw[very thick] (-\xR-2.5,2.1) arc[start angle=180,end angle=360,x radius=\xR,y radius=\yR]; 

\node at (-2.5,2.1) {\huge $\circledast_5$};

 \tvcl{(-2.5,2.1-\yR-0.1)}{f1b}{}{white}%
     \tvcl{(-2.5,2.1+\yR+0.1)}{f1t}{}{white}%
    
    \node at (-2,-0.2) {\LARGE $\ope{z}$};
    \draw [outedgecolor,ultra thick] (q1v4) -- (f1t); 

    \draw [outedgecolor,ultra thick] (q4v4) -- (f1b); 

    \node at (-2,4.3) {\LARGE $\ope{z}$};

\fill[gray!20] (11.5,2.1) ellipse ({\xR} and {\yR}); 
\draw[very thick] (-\xR+11.5,2.1) arc[start angle=180,end angle=0,x radius=\xR,y radius=\yR]; 
\draw[very thick] (-\xR+11.5,2.1) arc[start angle=180,end angle=360,x radius=\xR,y radius=\yR]; 

\node at (11.5,2.1) {\huge $\circledast_5$}; 

      \tvcl{(11.5,2.1-\yR-0.1)}{f2b}{}{white}%
     \tvcl{(11.5,2.1+\yR+0.1)}{f2t}{}{white}%
     
    \node at (10.1,-2) {\LARGE $\ope{z}$};
    \path [outedgecolor,ultra thick,bend right=90] (q6v1) edge (f2b); 

    \path [outedgecolor,ultra thick,bend left=90] (q3v3) edge (f2t); 

    \node at (10.1,6.1) {\LARGE $\ope{z}$};

\end{tikzpicture}
}
    \end{center}
\end{minipage}
\hspace{2cm}
\begin{minipage}{3.5cm}
    \begin{center}
    \textbf{(a2) $\mathbi{X}$-$\mathtt{TTF}$}

    \smallskip

    \scalebox{0.35}{\begin{tikzpicture}[scale=0.8]

    \node at (2,2) {\huge $\Theta_{\mathtt{F}}$};
    \node at (6,2) {\huge $\Theta_{\mathtt{T}}$};
    
    \vcl{(0,3)}{q1v1}{}
	\vcl{(1,4)}{q1v2}{}
    \vcl{(0,5)}{q1v3}{}
	\tvcl{(-1,4)}{q1v4}{}{black}
    \node at (0,4) {\LARGE $\mathsf{Q}^{\z\!X}_1$};

    \vcl{(4,3)}{q2v1}{}
	\vcl{(5,4)}{q2v2}{}
    \vcl{(4,5)}{q2v3}{}
	\vcl{(3,4)}{q2v4}{}
    \node at (4,4) {\LARGE $\mathsf{Q}^{\z\!X}_2$};

    \vcl{(8,3)}{q3v1}{}
	\vcl{(9,4)}{q3v2}{}
    \tvcl{(8,5)}{q3v3}{}{black}
	\vcl{(7,4)}{q3v4}{}
    \node at (8,4) {\LARGE $\mathsf{Q}^{\z\!X}_3$};
    
    \vcl{(0,-1)}{q4v1}{}
	\vcl{(1,0)}{q4v2}{}
    \vcl{(0,1)}{q4v3}{}
	\tvcl{(-1,0)}{q4v4}{}{black}
    \node at (0,0) {\LARGE $\mathsf{Q}^{\z\!X}_4$};
    
    \vcl{(4,-1)}{q5v1}{}
	\vcl{(5,0)}{q5v2}{}
    \vcl{(4,1)}{q5v3}{}
	\vcl{(3,0)}{q5v4}{}
    \node at (4,0) {\LARGE $\mathsf{Q}^{\z\!X}_5$};
    
    \tvcl{(8,-1)}{q6v1}{}{black}
	\vcl{(9,0)}{q6v2}{}
    \vcl{(8,1)}{q6v3}{}
	\vcl{(7,0)}{q6v4}{}
    \node at (8,0) {\LARGE $\mathsf{Q}^{\z\!X}_6$};

    \dted{q1v1}{q1v2}{orange}
    \ed{q1v1}{q1v4}{orange}
    \ed{q1v3}{q1v2}{orange}
    \dted{q1v3}{q1v4}{orange}
    \draw[outedgecolor,ultra thick] (q1v3) -- (0,6) node[anchor=west]{\LARGE $\ope{\mathtt{T}_1}$};
            
    \ed{q2v1}{q2v2}{orange}
    \dted{q2v1}{q2v4}{orange}
    \dted{q2v3}{q2v2}{orange}
    \ed{q2v3}{q2v4}{orange}
    \draw[outedgecolor,ultra thick] (q2v3) -- (4,6) node[anchor=west]{\LARGE $\ope{\mathtt{T}_1}$};

    \dted{q3v1}{q3v2}{orange}
    \ed{q3v1}{q3v4}{orange}
    \ed{q3v3}{q3v2}{orange}
    \dted{q3v3}{q3v4}{orange}
    \draw[outedgecolor,ultra thick] (q3v2) -- (10,4) node[anchor=west]{\LARGE $\ope{\mathtt{F}}$};

    \ed{q4v1}{q4v2}{orange}
    \dted{q4v1}{q4v4}{orange}
    \dted{q4v3}{q4v2}{orange}
    \ed{q4v3}{q4v4}{orange}
    \draw[outedgecolor,ultra thick] (q4v1) -- (0,-2) node[anchor=west]{\LARGE $\ope{\mathtt{T}_2}$};

    \dted{q5v1}{q5v2}{orange}
    \ed{q5v1}{q5v4}{orange}
    \ed{q5v3}{q5v2}{orange}
    \dted{q5v3}{q5v4}{orange}
     \draw[outedgecolor,ultra thick] (q5v1) -- (4,-2) node[anchor=west]{\LARGE $\ope{\mathtt{T}_2}$};

    \ed{q6v1}{q6v2}{orange}
    \dted{q6v1}{q6v4}{orange}
    \dted{q6v3}{q6v2}{orange}
    \ed{q6v3}{q6v4}{orange}
    \draw[outedgecolor,ultra thick] (q6v2) -- (10,0) node[anchor=west]{\LARGE $\ope{\mathtt{F}}$};
    
    \ed{q1v2}{q2v4}{intedgecolor}
    \ed{q1v1}{q4v3}{intedgecolor}
    \ed{q2v2}{q3v4}{intedgecolor}
    \ed{q2v1}{q5v3}{intedgecolor}
    \ed{q3v1}{q6v3}{intedgecolor}
    \ed{q4v2}{q5v4}{intedgecolor}
    \ed{q5v2}{q6v4}{intedgecolor}

\end{tikzpicture}
}
    \end{center}
 \end{minipage}
 \hspace{1.5cm}
 \begin{minipage}{4cm}

    \begin{center}
\textbf{(b) $\mathbi{X}$-$\mathtt{TF}$}

    \smallskip

    \scalebox{0.35}{\begin{tikzpicture}[scale=0.8]

    \node at (2,2) {\huge $\Theta_{\mathtt{F}}$};
    \node at (6,2) {\huge $\Theta_{\mathtt{T}}$};
    
    \vcl{(0,3)}{q1v1}{}
	\vcl{(1,4)}{q1v2}{}
    \vcl{(0,5)}{q1v3}{}
	\tvcl{(-1,4)}{q1v4}{}{black}
    \node at (0,4) {\LARGE $\mathsf{Q}^{\z\!X}_1$};

    \vcl{(4,3)}{q2v1}{}
	\vcl{(5,4)}{q2v2}{}
    \vcl{(4,5)}{q2v3}{}
	\vcl{(3,4)}{q2v4}{}
    \node at (4,4) {\LARGE $\mathsf{Q}^{\z\!X}_2$};

    \vcl{(8,3)}{q3v1}{}
	\vcl{(9,4)}{q3v2}{}
    \tvcl{(8,5)}{q3v3}{}{black}
	\vcl{(7,4)}{q3v4}{}
    \node at (8,4) {\LARGE $\mathsf{Q}^{\z\!X}_3$};
    
    \tvcl{(0,-1)}{q4v1}{}{black}
	\vcl{(1,0)}{q4v2}{}
    \vcl{(0,1)}{q4v3}{}
	\tvcl{(-1,0)}{q4v4}{}{black}
    \node at (0,0) {\LARGE $\mathsf{Q}^{\z\!X}_4$};
    
    \tvcl{(4,-1)}{q5v1}{}{black}
	\vcl{(5,0)}{q5v2}{}
    \vcl{(4,1)}{q5v3}{}
	\vcl{(3,0)}{q5v4}{}
    \node at (4,0) {\LARGE $\mathsf{Q}^{\z\!X}_5$};
    
    \tvcl{(8,-1)}{q6v1}{}{black}
	\vcl{(9,0)}{q6v2}{}
    \vcl{(8,1)}{q6v3}{}
	\vcl{(7,0)}{q6v4}{}
    \node at (8,0) {\LARGE $\mathsf{Q}^{\z\!X}_6$};

    \dted{q1v1}{q1v2}{orange}
    \ed{q1v1}{q1v4}{orange}
    \ed{q1v3}{q1v2}{orange}
    \dted{q1v3}{q1v4}{orange}
    \draw[outedgecolor,ultra thick] (q1v3) -- (0,6) node[anchor=west]{\LARGE $\ope{\mathtt{T}}$};
            
    \ed{q2v1}{q2v2}{orange}
    \dted{q2v1}{q2v4}{orange}
    \dted{q2v3}{q2v2}{orange}
    \ed{q2v3}{q2v4}{orange}
    \draw[outedgecolor,ultra thick] (q2v3) -- (4,6) node[anchor=west]{\LARGE $\ope{\mathtt{T}}$};

    \dted{q3v1}{q3v2}{orange}
    \ed{q3v1}{q3v4}{orange}
    \ed{q3v3}{q3v2}{orange}
    \dted{q3v3}{q3v4}{orange}
    \draw[outedgecolor,ultra thick] (q3v2) -- (10,4) node[anchor=west]{\LARGE $\ope{\mathtt{F}}$};

    \ed{q4v1}{q4v2}{orange}
    \dted{q4v1}{q4v4}{orange}
    \dted{q4v3}{q4v2}{orange}
    \ed{q4v3}{q4v4}{orange}

    \dted{q5v1}{q5v2}{orange}
    \ed{q5v1}{q5v4}{orange}
    \ed{q5v3}{q5v2}{orange}
    \dted{q5v3}{q5v4}{orange}

    \ed{q6v1}{q6v2}{orange}
    \dted{q6v1}{q6v4}{orange}
    \dted{q6v3}{q6v2}{orange}
    \ed{q6v3}{q6v4}{orange}
    \draw[outedgecolor,ultra thick] (q6v2) -- (10,0) node[anchor=west]{\LARGE $\ope{\mathtt{F}}$};
    
    \ed{q1v2}{q2v4}{intedgecolor}
    \ed{q1v1}{q4v3}{intedgecolor}
    \ed{q2v2}{q3v4}{intedgecolor}
    \ed{q2v1}{q5v3}{intedgecolor}
    \ed{q3v1}{q6v3}{intedgecolor}
    \ed{q4v2}{q5v4}{intedgecolor}
    \ed{q5v2}{q6v4}{intedgecolor}

\end{tikzpicture}
}
    
   \vspace{3mm}
  \end{center}
 \end{minipage}
 
        \medskip

\begin{minipage}{4cm}
    \begin{center}
      \textbf{(c) $\mathbi{W}$}

      \vspace{5mm}
       
      \scalebox{0.35}{
\begin{tikzpicture}[scale=0.8]

    \vcl{(1,0)}{q1v1}{}
	\tvcl{(2.4,0)}{q1v2}{}{black}
    \vcl{(2.4,1.4)}{q1v3}{}
	\vcl{(1,1.4)}{q1v4}{}

    \vcl{(1,3)}{q2v1}{}
	\vcl{(2.4,3)}{q2v2}{}
    \tvcl{(2.4,4.4)}{q2v3}{}{black}
	\vcl{(1,4.4)}{q2v4}{}
   
    \dted{q1v1}{q1v2}{orange} 
    \ed{q1v1}{q1v4}{orange}
    \ed{q1v3}{q1v2}{orange}
    \dted{q1v3}{q1v4}{orange}
    \draw[outedgecolor,ultra thick] (q1v1) -- (0,0) node[anchor=east]{\LARGE $\ope{X}$};
    \draw[outedgecolor,ultra thick] (q1v3) -- (3.4,1.4) node[anchor=west]{\LARGE $\ope{Y}$};
            
    \dted{q2v1}{q2v2}{orange}
    \ed{q2v1}{q2v4}{orange}
    \ed{q2v3}{q2v2}{orange}
    \dted{q2v3}{q2v4}{orange}
    \draw[outedgecolor,ultra thick] (q2v4) -- (0,4.4) node[anchor=east]{\LARGE $\ope{X}$};
    \draw[outedgecolor,ultra thick] (q2v2) -- (3.4,3) node[anchor=west]{\LARGE $\ope{Y}$};
            
    \ed{q1v4}{q2v1}{intedgecolor}

\end{tikzpicture}
}

      \vspace{8mm}
    \end{center}  
\end{minipage}
\hspace{2cm}
\begin{minipage}{3.5cm}
    \begin{center}
    \textbf{(d) $2$-\!$\mathbi{Y}$}

    \vspace{5mm}
    
    \hspace{-1.2cm}\scalebox{0.35}{
\begin{tikzpicture}[scale=0.8]

    \node at (2,2) {\huge $\Theta_1$};
    \node at (6,2) {\huge $\Theta_2$};

    \vcl{(0,3)}{q1v1}{}
	\vcl{(1,4)}{q1v2}{}
    \tvcl{(0,5)}{q1v3}{}{black}
	\vcl{(-1,4)}{q1v4}{}
    \node at (0,4) {\LARGE $\mathsf{Q}^{\z\!Y}_1$};
    
    \vcl{(4,3)}{q2v1}{}
	\vcl{(5,4)}{q2v2}{}
    \tvcl{(4,5)}{q2v3}{}{black}
	\vcl{(3,4)}{q2v4}{}
    \node at (4,4) {\LARGE $\mathsf{Q}^{\z\!Y}_2$};
    
    \vcl{(8,3)}{q3v1}{}
	\tvcl{(9,4)}{q3v2}{}{black}
    \tvcl{(8,5)}{q3v3}{}{black}
	\vcl{(7,4)}{q3v4}{}
    \node at (8,4) {\LARGE $\mathsf{Q}^{\z\!Y}_3$};

    \tvcl{(0,-1)}{q4v1}{}{black}
	\vcl{(1,0)}{q4v2}{}
    \vcl{(0,1)}{q4v3}{}
	\vcl{(-1,0)}{q4v4}{}
    \node at (0,0) {\LARGE $\mathsf{Q}^{\z\!Y}_4$};
    
    \vcl{(4,-1)}{q5v1}{}
	\vcl{(5,0)}{q5v2}{}
    \vcl{(4,1)}{q5v3}{}
	\vcl{(3,0)}{q5v4}{}
    \node at (4,0) {\LARGE $\mathsf{Q}^{\z\!Y}_5$};
    
    \vcl{(8,-1)}{q6v1}{}
	\tvcl{(9,0)}{q6v2}{}{black}
    \vcl{(8,1)}{q6v3}{}
	\vcl{(7,0)}{q6v4}{}
    \node at (8,0) {\LARGE $\mathsf{Q}^{\z\!Y}_6$};

    \ed{q1v1}{q1v2}{orange}
    \dted{q1v1}{q1v4}{orange}
    \dted{q1v3}{q1v2}{orange}
    \ed{q1v3}{q1v4}{orange}
    \draw[outedgecolor,ultra thick] (q1v4) -- (-2,4) node[anchor=east]{\LARGE $\ope{W_1}$};
            
    \dted{q2v1}{q2v2}{orange}
    \ed{q2v1}{q2v4}{orange}
    \ed{q2v3}{q2v2}{orange}
    \dted{q2v3}{q2v4}{orange}

    \ed{q3v1}{q3v2}{orange}
    \dted{q3v1}{q3v4}{orange}
    \dted{q3v3}{q3v2}{orange}
    \ed{q3v3}{q3v4}{orange}

    \dted{q4v1}{q4v2}{orange}
    \ed{q4v1}{q4v4}{orange}
    \ed{q4v3}{q4v2}{orange}
    \dted{q4v3}{q4v4}{orange}
    \draw[outedgecolor,ultra thick] (q4v4) -- (-2,0) node[anchor=east]{\LARGE $\ope{W_1}$};

    \ed{q5v1}{q5v2}{orange}
    \dted{q5v1}{q5v4}{orange}
    \dted{q5v3}{q5v2}{orange}
    \ed{q5v3}{q5v4}{orange}
    \draw[outedgecolor,ultra thick] (q5v1) -- (4,-2) node[anchor=west]{\LARGE $\ope{W_2}$};

    \dted{q6v1}{q6v2}{orange}
    \ed{q6v1}{q6v4}{orange}
    \ed{q6v3}{q6v2}{orange}
    \dted{q6v3}{q6v4}{orange}
    \draw[outedgecolor,ultra thick] (q6v1) -- (8,-2) node[anchor=west]{\LARGE $\ope{W_2}$};
    
    \ed{q1v2}{q2v4}{intedgecolor}
    \ed{q1v1}{q4v3}{intedgecolor}
    \ed{q2v2}{q3v4}{intedgecolor}
    \ed{q2v1}{q5v3}{intedgecolor}
    \ed{q3v1}{q6v3}{intedgecolor}
    \ed{q4v2}{q5v4}{intedgecolor}
    \ed{q5v2}{q6v4}{intedgecolor}

\end{tikzpicture}
}
    \end{center}
 \end{minipage}
 \hspace{1.5cm}
 \begin{minipage}{4cm}
    \begin{center}
    \textbf{(e) $3$-\!$\mathbi{Y}$}\\

     \smallskip
     
    \scalebox{0.35}{
\begin{tikzpicture}[scale=0.8]

    \node at (2,2) {\huge $\Theta_1$};
    \node at (6,2) {\huge $\Theta_2$};
    \node at (10,2) {\huge $\Theta_3$};

    \vcl{(0,3)}{q1v1}{}
	\vcl{(1,4)}{q1v2}{}
    \vcl{(0,5)}{q1v3}{}
	\vcl{(-1,4)}{q1v4}{}
    \node at (0,4) {\LARGE $\mathsf{Q}^{\z\!Y}_1$};
    
    \vcl{(4,3)}{q2v1}{}
	\vcl{(5,4)}{q2v2}{}
    \vcl{(4,5)}{q2v3}{}
	\vcl{(3,4)}{q2v4}{}
    \node at (4,4) {\LARGE $\mathsf{Q}^{\z\!Y}_2$};
    
    \vcl{(8,3)}{q3v1}{}
	\vcl{(9,4)}{q3v2}{}
    \vcl{(8,5)}{q3v3}{}
	\vcl{(7,4)}{q3v4}{}
    \node at (8,4) {\LARGE $\mathsf{Q}^{\z\!Y}_3$};

    \vcl{(0,-1)}{q4v1}{}
	\vcl{(1,0)}{q4v2}{}
    \vcl{(0,1)}{q4v3}{}
	\vcl{(-1,0)}{q4v4}{}
    \node at (0,0) {\LARGE $\mathsf{Q}^{\z\!Y}_4$};
    
    \vcl{(4,-1)}{q5v1}{}
	\vcl{(5,0)}{q5v2}{}
    \vcl{(4,1)}{q5v3}{}
	\vcl{(3,0)}{q5v4}{}
    \node at (4,0) {\LARGE $\mathsf{Q}^{\z\!Y}_5$};
    
    \vcl{(8,-1)}{q6v1}{}
	\vcl{(9,0)}{q6v2}{}
    \vcl{(8,1)}{q6v3}{}
	\vcl{(7,0)}{q6v4}{}
    \node at (8,0) {\LARGE $\mathsf{Q}^{\z\!Y}_6$};

    \ed{q1v1}{q1v2}{orange}
    \dded{q1v1}{q1v4}{orange}
    \dded{q1v3}{q1v2}{orange}
    \ed{q1v3}{q1v4}{orange}
    \draw[outedgecolor,ultra thick] (q1v3) -- (0,6) node[anchor=west]{\LARGE $\ope{W_1}$};
            
    \dted{q2v1}{q2v2}{orange}
    \ed{q2v1}{q2v4}{orange}
    \ed{q2v3}{q2v2}{orange}
    \dted{q2v3}{q2v4}{orange}
    \draw[outedgecolor,ultra thick] (q2v3) -- (4,6) node[anchor=west]{\LARGE $\ope{W_1}$};

    \dded{q3v1}{q3v2}{orange}
    \dted{q3v1}{q3v4}{orange}
    \dted{q3v3}{q3v2}{orange}
    \dded{q3v3}{q3v4}{orange}
    \draw[outedgecolor,ultra thick] (q3v3) -- (8,6) node[anchor=west]{\LARGE $\ope{W_3}$};

    \dded{q4v1}{q4v2}{orange}
    \ed{q4v1}{q4v4}{orange}
    \ed{q4v3}{q4v2}{orange}
    \dded{q4v3}{q4v4}{orange}
    \draw[outedgecolor,ultra thick] (q4v1) -- (0,-2) node[anchor=west]{\LARGE $\ope{W_3}$};

    \ed{q5v1}{q5v2}{orange}
    \dted{q5v1}{q5v4}{orange}
    \dted{q5v3}{q5v2}{orange}
    \ed{q5v3}{q5v4}{orange}
    \draw[outedgecolor,ultra thick] (q5v1) -- (4,-2) node[anchor=west]{\LARGE $\ope{W_2}$};

    \dted{q6v1}{q6v2}{orange}
    \dded{q6v1}{q6v4}{orange}
    \dded{q6v3}{q6v2}{orange}
    \dted{q6v3}{q6v4}{orange}
    \draw[outedgecolor,ultra thick] (q6v1) -- (8,-2) node[anchor=west]{\LARGE $\ope{W_2}$};
    
    \ed{q1v2}{q2v4}{intedgecolor}
    \ed{q1v1}{q4v3}{intedgecolor}
    \ed{q2v2}{q3v4}{intedgecolor}
    \ed{q2v1}{q5v3}{intedgecolor}
    \ed{q3v1}{q6v3}{intedgecolor}
    \ed{q4v2}{q5v4}{intedgecolor}
    \ed{q5v2}{q6v4}{intedgecolor}

    \draw[intedgecolor,ultra thick] (q3v2) -- (11.5,4) -- (11.5,-2.5) -- (-2,-2.5) -- (-2, 0) -- (q4v4);

    \draw[intedgecolor,ultra thick] (q1v4) -- (-2, 4) -- (-2, 6.5) -- (11,6.5) -- (11,0) -- (q6v2);

\end{tikzpicture}
}
    
    \end{center}
 \end{minipage}

   \end{center}
   
    \caption{\label{fig:gadgets}
      \figcbf{(a1-b)} Each variable gadget $X$ has two competing 8-cycles: $\Theta_{\mathtt{F}}$ is the active 8-cycle when resolving squares $\mathsf{Q}^{\!X}_1$, $\mathsf{Q}^{\!X}_2$, $\mathsf{Q}^{\!X}_4$ and $\mathsf{Q}^{\!X}_5$ using the dotted edges and $\Theta_{\mathtt{T}}$ is the active 8-cycle when resolving squares $\mathsf{Q}^{\!X}_2$, $\mathsf{Q}^{\!X}_3$, $\mathsf{Q}^{\!X}_5$ and $\mathsf{Q}^{\!X}_6$ using the solid edges.
      The gadget has one or two pairs of edges with label $\ope{\mathtt{T}_{(i)}}$, compatible with cycle $\Theta_{\mathtt{T}}$, and a pair of edges with label $\ope{\mathtt{F}}$, compatible with cycle $\Theta_{\mathtt{F}}$. These pairs will be used to connect 
      the corresponding literal occurrences in clauses.
      \figcbf{(a1)} Gadget for a variable $x$ that occurs twice as the corresponding positive literal and once as the negative one. The four edges labeled with $\langle z\rangle$ are connected to two open 5-flowers.  
      \figcbf{(a2)} Simplified picture of case (a1), where the open 5-flowers are omitted and the four vertices connected to them are represented in black. 
      \figcbf{(b)} Gadget for a variable $x$ that occurs once as the corresponding positive literal and once as the negative one. The six black vertices are connected to three (omitted) open 5-flowers. 
      \figcbf{(c)} Graph obtained for each literal occurrence to connect the occurrence in a clause to the graph representing the value of the variable.
      The edges with the label $\ope{X}$ connect to the variable gadget while the edges with the label $\ope{Y}$ connect to the clause gadget where the corresponding literal occurrence is.
      The two black vertices are connected to one (omitted) open 5-flower.
      \figcbf{(d)} Graph obtained for each 2-clause.
      In the center we have two competing 8-cycles: $\Theta_1$ is the active 8-cycle when 
      resolving squares $\mathsf{Q}^{\!Y}_1$, $\mathsf{Q}^{\!Y}_2$, $\mathsf{Q}^{\!Y}_4$ and $\mathsf{Q}^{\!Y}_5$ using the solid edges; while $\Theta_2$ is the active 8-cycle when resolving squares $\mathsf{Q}^{\!Y}_2$, $\mathsf{Q}^{\!Y}_3$, $\mathsf{Q}^{\!Y}_5$ and $\mathsf{Q}^{\!Y}_6$ using the dotted edges.
      The pairs of edges with the same label $\ope{W_1}$ or $\ope{W_2}$ represent the pairs to connect the  corresponding literal gadgets. 
      The six black vertices are connected to three (omitted) open 5-flowers.
      \figcbf{(e)} Graph obtained for each 3-clause.
      In the center we have three competing 8-cycles: in addition to $\Theta_1$ and $\Theta_2$ described in case (d), here $\Theta_3$ is the active 8-cycle when 
      resolving squares $\mathsf{Q}^{\!Y}_1$, $\mathsf{Q}^{\!Y}_3$, $\mathsf{Q}^{\!Y}_4$ and $\mathsf{Q}^{\!Y}_6$ using the dashed edges. 
      The pairs of edges with the same label $\ope{W_1}$, $\ope{W_2}$ or $\ope{W_3}$ represent the pairs to connect the corresponding literal gadgets. 
    }
\end{figure*}

\noindent
For $\mathcal{X}=\{x_1, x_2, x_3, x_4\}$ and 
$\mathcal{Y} = [ (x_1 \vee x_2 ) \wedge (x_1 \vee x_3) \wedge (\overline{x_1} \vee \overline{x_2} \vee x_4) \wedge (x_2 \vee x_3) \wedge (\overline{x_3} \vee \overline{x_4}) ]$, the graph $G_8(\mathcal{X},\mathcal{Y})$
is shown in Fig.~\ref{fig:connstexample}~(a).
The numbers of squares per variable, clause and literal occurrence are clearly bounded by constant values, therefore Construction~\ref{const:3satAmb} is polynomial in $\|\mathcal{Y}\|$.

\begin{figure*}[ht!]
  \begin{center}
  \begin{minipage}{5cm}
    \centering
    \textbf{(a)}\\[2ex]
    \input{Fig-example}
  \end{minipage}
  \hspace{30mm}  
  \begin{minipage}{5cm}
    \centering
    \textbf{(b)}\\[2ex]
    \input{Fig-completed}
  \end{minipage}
  \end{center}
  \caption{\label{fig:connstexample}
    Consider the $\satvar$-SAT formula $\mathcal{Y} = [ (x_1 \vee x_2 ) \wedge (x_1 \vee x_3) \wedge (\overline{x_1} \vee \overline{x_2} \vee x_4) \wedge (x_2 \vee x_3) \wedge (\overline{x_3} \vee \overline{x_4}) ]$, over $\mathcal{X}=\{x_1,x_2,x_3,x_4\}$. 
    Note that $|\mathcal{X}|=4$, $|\mathcal{Y}|=5$ and $\|\mathcal{Y}\|=11$.
    \figcbf{(a)} Graph $G_8(\mathcal{X},\mathcal{Y})$, 
    obtained with Construction~\ref{const:3satAmb}. 
    The variable gadgets are on the left, and their labels make it explicit to which variable each gadget refers.
    The clause gadgets are on the right, and again their labels make it explicit to which clause each gadget refers.
    The connecting black edges referring to each variable and its literals are as follows: for $x_1$ they are dotted, for $x_2$ they are solid, for $x_3$ they are dashed and for $x_4$ they are dash-dotted. This facilitates the visualization of the literal connecting gadgets in the center.
    All open 5-flowers are omitted and all the vertices connected to them are represented in black. 
    (Once the parity is fulfilled, it does not matter which pair of black vertices of the graph are selected to complete each flower.) 
    \figcbf{(b)} The truth assignment $A=\{ (x_1=\mathtt{T}),(x_2=\mathtt{T}),(x_3=\mathtt{F}),(x_4=\mathtt{T})\}$ gives an optimal solution of $G_8(\mathcal{X},\mathcal{Y})$ whose 8-score is $|\mathcal{X}|+|\mathcal{Y}|+\|\mathcal{Y}\|=4+5+11=20$. The edges of 8-cycles are highlighted, while the edges that are in longer cycles (completing open 5-flowers) are shaded. Here we chose $x_2$ as the literal responsible for satisfying the first clause of $\mathcal{Y}$, but choosing $x_1$ instead would be a valid alternative that gives another optimal solution of $G_8(\mathcal{X},\mathcal{Y})$ (see Fig.\;\ref{fig:connstexample-alt}~(a)). 
  }  
\end{figure*}

\begin{figure*}
  \begin{center}

   \begin{minipage}{5cm}
   \centering
   \textbf{(a)}\\[1ex]
 \input{Fig-completed-alt}
  \end{minipage}
  \hspace{30mm}  
   \begin{minipage}{5cm}
   \centering
   \textbf{(b)}\\[1ex]
 \input{Fig-completed-no-A}
  \end{minipage}

  \end{center}
  
  \caption{\label{fig:connstexample-alt}
    \figcbf{(a)} Alternative optimal solution for the truth assignment $A=\{ (x_1=\mathtt{T}),(x_2=\mathtt{T}),(x_3=\mathtt{F}),(x_4=\mathtt{T})\}$ for the example given in Fig.\;\ref{fig:connstexample}. Here, instead of $x_2$, we chose $x_1$ as the literal responsible for satisfying the first clause of $\mathcal{Y}$, obtaining 20 8-cycles.
    \figcbf{(b)} Assignment $A'=\{ (x_1=\mathtt{T}),(x_2=\mathtt{F}),(x_3=\mathtt{F}),(x_4=\mathtt{T})\}$ for the same example does not satisfy the fourth clause of $\mathcal{Y}$. By fixing $A'$ in $G_8(\mathcal{X},\mathcal{Y})$, the best solution has 19 8-cycles.
  }  
\end{figure*}

One crucial feature of the design of Construction~\ref{const:3satAmb} is that it avoids alternating cycles of lengths up to 6.
Let $G_\bullet(\mathcal{X},\mathcal{Y})$ be the graph obtained from $G_8(\mathcal{X},\mathcal{Y})$ by removing all 5-flowers.

\begin{proposition}\label{prop:no-cycle-up-to-6}
The graph $G_\bullet(\mathcal{X},\mathcal{Y})$, and consequently also $G_8(\mathcal{X},\mathcal{Y})$, does not contain as an induced subgraph any alternating cycle whose length is smaller than 8.
\end{proposition}
\begin{proof}
Since $G_\bullet$ has no black edge directly connecting vertices of the same square, no 2-cycle is possible. Furthermore, since $G_\bullet$ has no pair of squares connected to each other by two or more black edges, no 4-cycle is possible. Finally, since $G_\bullet$ has no set of three squares such that there are black edges connecting each pair of them, no 6-cycle is possible.
Recall that open 5-flowers only admit alternating cycles of length at least 10, therefore the observations above hold for $G_8(\mathcal{X},\mathcal{Y})$ as well.
\end{proof}

The graph $G_8(\mathcal{X},\mathcal{Y})$ consists of orange squares and black connections, such that every vertex has degree three, which are the properties of an ambiguous breakpoint graph of circular genomes.

\begin{proposition}\label{prop:ABG-circ}
There exists a pair of $\cogsd$-cognate circular genomes $\mathbb{S}$ and $\mathbb{D}$ and a genome $\check{\mathbb{D}} \in \mathfrak{S}^\mathtt{a}_\mathtt{b}(\mathbb{D})$, such that  $G_8(\mathcal{X},\mathcal{Y})=ABG(\mathbb{S},\check{\mathbb{D}})$.
\end{proposition}
\begin{proof}
Let $n_\ast$ be the number of squares in $G_8(\mathcal{X},\mathcal{Y}) = (V,E)$. (The number of squares is equivalent to the number of adjacencies and to the number of genes in any singular circular genome that is derived from the graph.)
Clearly, $|V| = 4n_\ast$.
The set of edges $E$ can be partitioned into two sets: $E_{\mathbb{S}}$, containing all edges that originate from the orange squares, and $E_{\mathbb{D}}$, containing the black connections.
Now let the set of gene identifiers be $I=\{\mathtt{1,2,3},\ldots,n_\ast\}$.
The set of extremities is then $\Gamma= \bigcup_{\mathtt{X} \in I} \{\mathtt{X}^h,\mathtt{X}^t\}$. Note that $|\Gamma|=2n_\ast$ is even and let $A$ be any perfect matching of $\Gamma$. 
Since $|A|=\frac{|\Gamma|}{2}=n_\ast$, we can take the set $A$ as the adjacencies of the singular genome $\mathbb{S}$.
We then assign each pair $\beta\gamma \in A$ to a distinct square $Q=\mathcal{Q}(\beta\gamma)$ of $G$ by labeling the vertices of $Q$ as follows: take any edge of $Q$ and call its extremities $u$ and $v$;
then assign $u \rightarrow \beta_\mathtt{a}$, $\hat{u} \rightarrow \beta_\mathtt{b}$, $v \rightarrow \gamma_\mathtt{a}$, $\hat{v} \rightarrow \gamma_\mathtt{b}$.
The adjacencies of $\check{\mathbb{D}}$ are then the edge set $E_{\mathbb{D}}$; and the duplicated genome $\mathbb{D}$ can be obtained from $\check{\mathbb{D}}$ by simply ignoring the singularization indices. 
\end{proof}

This guarantees that solving the $\sigma_8$ disambiguation or finding an 8-maximum solution of a graph $G_8(\mathcal{X},\mathcal{Y})$ is equivalent to solving the $\sigma_8$ double distance of the corresponding circular genomes $\mathbb{S}$ and $\mathbb{D}$.
Now we will determine the relation between the solution of the $\satvar$-SAT instance $(\mathcal{X},\mathcal{Y})$ and an 8-maximum solution of $G_8(\mathcal{X},\mathcal{Y})$.

\begin{lemma}\label{lmm:iff-circ}
For any 8-maximum solution $\tau$ of graph $G_8(\mathcal{X},\mathcal{Y})$, it holds that:
(i) $\textup{s}_8(\tau) \leq |\mathcal{X}|+|\mathcal{Y}|+\|\mathcal{Y}\|$;
and (ii) formula $\mathcal{Y}$ has a satisfying assignment if and only if $\textup{s}_8(\tau) = |\mathcal{X}|+|\mathcal{Y}|+\|\mathcal{Y}\|$.
\end{lemma}

\begin{proof}
The smallest cycle in $G_8(\mathcal{X},\mathcal{Y})$ has length 8 (Proposition~\ref{prop:no-cycle-up-to-6}), therefore the 8-score of any solution simply corresponds to its number of 8-cycles.
An 8-maximum solution of $G_\bullet(\mathcal{X},\mathcal{Y})$ is also an 8-maximum solution of $G_8(\mathcal{X},\mathcal{Y})$.

\smallskip

\noindent (i) The first statement can be easily verified as follows.
For each variable gadget $X$ (Fig.\;\ref{fig:gadgets}~(a1-a2) and~(b)) only one of the two competing 8-cycles can be induced.
For each clause gadget $Y$ (Fig.\;\ref{fig:gadgets}~(d) and~(e)) only one of the two or three competing 8-cycles can be induced.  
Each literal occurrence gadget $W$ (Fig.\;\ref{fig:gadgets}~(c)) can either form an 8-cycle with a clause gadget $Y$, or an 8-cycle with a variable gadget $X$. 
Therefore, together all the gadgets give at most $|\mathcal{X}|+|\mathcal{Y}|+\|\mathcal{Y}\|$ 8-cycles.

\medskip

\noindent (ii) For proving the second statement, we first show that from a satisfying assignment $A$ of $\mathcal{Y}$ we can obtain a solution $\tau$ of $G_\bullet(\mathcal{X},\mathcal{Y})$ composed of $|\mathcal{X}|+|\mathcal{Y}|+\|\mathcal{Y}\|$ $8$-cycles.

\begin{itemize}
  \item
    For each variable $x$: if $x=\mathtt{T}$ in $A$, select the edges inducing 8-cycle $\Theta_{\mathtt{T}}$ in the respective gadget $X$ (Fig.\;\ref{fig:gadgets}~(a1-a2) and~(b)). 
    Otherwise, if $x=\mathtt{F}$ in $A$, select the edges inducing 8-cycle $\Theta_{\mathtt{F}}$.
    In any case, each variable gadget will give a single $8$-cycle, therefore in total we increase the score by~$|\mathcal{X}|$.
    
    Furthermore, if cycle $\Theta_\mathtt{T}$ is induced, the edges of the remaining squares must induce path(s) of length 3 connecting $\ope{\mathtt{T}_{(i)}}$ to $\ope{\mathtt{T}_{(i)}}$ allowing the formation of 8-cycle(s) with the corresponding literal gadget(s); if cycle $\Theta_\mathtt{F}$ is induced, the edges of the remaining squares must induce a path of length 3 connecting $\ope{\mathtt{F}}$ to $\ope{\mathtt{F}}$ allowing the formation of an 8-cycle with the corresponding literal gadget. How these connections are done is explained in the third item below.
  \item
    For each clause $y$: pick one literal $w$ of clause $y$ that is satisfied by $A$ (there can be multiple choices). If $w$ is the first literal, select the edges inducing 8-cycle $\Theta_1$ in clause gadget $Y$ (Fig.\;\ref{fig:gadgets}~(d) or (e)). 
    Otherwise, if $w$ is the second literal, select the edges inducing 8-cycle $\Theta_2$. Finally, if $Y$ is a 3-clause and $w$ is its third literal, select the edges inducing 8-cycle $\Theta_3$.
    In any case, each clause gadget will give a single $8$-cycle, therefore in total we increase the score by~$|\mathcal{Y}|$.

    The two remaining squares should be resolved such that they induce path(s) of length 3 connecting $\ope{W_i}$ to $\ope{W_i}$, for each (not induced) cycle $\Theta_i$ corresponding to a literal that is not responsible for satisfying clause $y$.
    These 3-paths allow the formation of 8-cycles with the corresponding literal gadget(s).
    These connections are described in the next item.
  \item
    For each literal occurrence $w$: if $w$ is responsible to satisfy its clause $y$, select in the corresponding gadget~$W$ (Fig.\;\ref{fig:gadgets}~(c)) the edges inducing a path of length 3 connecting $\ope{X}$ to $\ope{X}$.
    This path induces an 8-cycle in the connection between $W$ and $X$ because cycle $\Theta_\mathtt{T}$ in $X$ is compatible with outer connections $\ope{\mathtt{T}_{(i)}}$, while cycle $\Theta_\mathtt{F}$ in $X$ is compatible with outer connections $\ope{\mathtt{F}}$.
    If $w$ is not responsible for satisfying $y$, select the edges inducing a path of length 3 connecting $\ope{Y}$ to $\ope{Y}$. This path induces an 8-cycle in the connection between $W$ and $Y$ because the cycle corresponding to $w$ in $Y$ is not induced.
    In this way, each literal occurrence gadget gives one $8$-cycle. The score is then increased by~$\|\mathcal{Y}\|$.
\end{itemize}

\noindent This procedure gives exactly $|\mathcal{X}|+|\mathcal{Y}|+\|\mathcal{Y}\|$ 8-cycles, and the remaining unresolved squares can be resolved arbitrarily. See examples in Fig.\;\ref{fig:connstexample}~(b) and Fig.\;\ref{fig:connstexample-alt}~(a). 

\smallskip 

Now, let $\tau$ be an 8-maximum solution of $G_\bullet(\mathcal{X},\mathcal{Y})$. If $\textup{s}_8(\tau)=|\mathcal{X}|+|\mathcal{Y}|+\|\mathcal{Y}\|$, since this is the maximum number of cycles possible, $\tau$ must have an 8-cycle per gadget $W$, and this covers the possible cycles that can be obtained with the outer connections of all gadgets. Furthermore, $\tau$ must include one of the two competing cycles $\Theta_\mathtt{T}$ or $\Theta_\mathtt{F}$ of each variable gadget and one of the two or three competing cycles $\Theta_1$ or $\Theta_2$ (or $\Theta_3$) of each clause gadget. The included $\Theta_*$ cycles then determine the value of each variable, giving an assignment $A$ of $\mathcal{X}$, and also determine, for each clause, one literal that satisfies it. The fact that each gadget $W$ composes an 8-cycle  
guarantees the compatibility of $A$. In other words, in this case $A$ satisfies $\mathcal{Y}$.
On the other hand, if $\textup{s}_8(\tau)<|\mathcal{X}|+|\mathcal{Y}|+\|\mathcal{Y}\|$, either $\tau$ gives an incomplete assignment or it gives an incompatible assignment that does not satisfy the formula $\mathcal{Y}$ (an intuition is given in Fig~\ref{fig:connstexample-alt}~(b)).
\end{proof}

\subsection{Reducing \boldmath$\satvar$-SAT to the \boldmath$\sigma_8$ disambiguation of linear genomes}

We can now extend the results from the previous section to linear genomes.
For that, we simply do a small extension of Construction~\ref{const:3satAmb}, by duplicating the number of vertices in the graph. 
We need to determine the number of vertices in 
$G_8(\mathcal{X},\mathcal{Y})$, which we denote by $\nu_8$.



Let $\mathcal{X}_\mathtt{TTF}$ be the subset of variables that occur three times in $\mathcal{Y}$, $\mathcal{X}_\mathtt{TF}$ be the subset of variables that occur twice in $\mathcal{Y}$, $\mathcal{Y}_2$ be the subset of 2-clauses in $\mathcal{Y}$ and $\mathcal{Y}_3$ be the subset of 3-clauses in $\mathcal{Y}$. 
Recall that each 5-flower has 5 squares, and, consequently, 20 vertices.
By inspecting the gadgets in Fig.\;\ref{fig:gadgets}, it is straightforward to see that the number of vertices is:
\begin{align*}
\nu_8=&~(24+2\times 20)|\mathcal{X}_\mathtt{TTF}|+(24+3\times 20)|\mathcal{X}_\mathtt{TF}|\\&~+(24+3\times 20)|\mathcal{Y}_2|+24|\mathcal{Y}_3|+(8+20)\|\mathcal{Y}\|\\
=&~64|\mathcal{X}|+20|\mathcal{X}_\mathtt{TF}|+24|\mathcal{Y}|+60|\mathcal{Y}_2|+28\|\mathcal{Y}\|\,.\end{align*}


\begin{construction}\label{const:3satAmbLin}
Given an instance $(\mathcal{X},\mathcal{Y})$ of $\satvar$-SAT, create the graph $G_8^{+\nu}(\mathcal{X},\mathcal{Y})$ by simply adding $\nu_8$ isolated vertices to the graph $G_8(\mathcal{X},\mathcal{Y})$.
\end{construction}

The vertices added to graph $G_8^{+\nu}(\mathcal{X},\mathcal{Y})$ must be telomeres of linear chromosomes in both genomes. In fact, it is easy to derive a pair of linear genomes from $G_8^{+\nu}(\mathcal{X},\mathcal{Y})$:

\begin{proposition}~\label{prop:ABG-lin}
There exists a pair of $\cogsd$-cognate linear genomes $\mathbb{S}$ and $\mathbb{D}$ and a genome $\check{\mathbb{D}} \in \mathfrak{S}^\mathtt{a}_\mathtt{b}(\mathbb{D})$, such that  $G_8^{+\nu}(\mathcal{X},\mathcal{Y})=ABG(\mathbb{S},\check{\mathbb{D}})$.
\end{proposition}

\begin{proof}
Let $G_8^{+\nu}(\mathcal{X},\mathcal{Y}) = (V,E)$ and recall that $|V| = 2\nu_8$. Note that $\nu_8$ is a multiple of 4, therefore we can write $\nu_8=4q$. 
Now let $n_\ast=2q$ and let the set of gene identifiers be $I=\{\mathtt{1,2,3},\ldots,n_\ast\}$.
The extremities are $\Gamma^t = \{\mathtt{X}^t \mid  \mathtt{X} \in I\}$ and $\Gamma^h = \{\mathtt{X}^h \mid  \mathtt{X} \in I\}$.  

Let $\Gamma^t$ be the telomeres of singular genome $\mathbb{S}$. Then assign each extremity $\mathtt{Z}^t \in \Gamma^t$ to a distinct pair $(u,v)$ of isolated vertices as follows: $u \rightarrow \mathtt{Z}^t_\mathtt{a}$ and $v \rightarrow\mathtt{Z}^t_\mathtt{b}$.
The telomeres of $\check{\mathbb{D}}$ are the isolated vertices of $G^{+\nu}$.

Note that $|\Gamma^h|=n_\ast=2q$ is even and let $A$ be any perfect matching of $\Gamma^h$. 
Take the set $A$ as the adjacencies of singular genome $\mathbb{S}$.
We then assign each pair $\mathtt{X}^h\mathtt{Y}^h \in A$ to a distinct square $Q=\mathcal{Q}(\mathtt{X}^h\mathtt{Y}^h)$ of $G^{+\nu}$ by labeling the vertices of $Q$ as follows: take any edge of $Q$ and call its extremities $u$ and $v$;
then assign $u \rightarrow \mathtt{X}^h_\mathtt{a}$, $\hat{u} \rightarrow \mathtt{X}^h_\mathtt{b}$, $v \rightarrow \mathtt{Y}^h_\mathtt{a}$, $\hat{v} \rightarrow \mathtt{Y}^h_\mathtt{b}$.
Now, let $E_{\mathbb{D}}$ be the subset of $E$ containing the black connections and let the adjacencies of $\check{\mathbb{D}}$ be the edge set $E_{\mathbb{D}}$.

Finally, the duplicated genome $\mathbb{D}$ can be obtained from $\check{\mathbb{D}}$ by simply ignoring the singularization indices. 

Since all gene tails are telomeres and all gene heads form adjacencies in both $\mathbb{D}$ and $\mathbb{S}$, no circular chromosome can occur in  $\mathbb{D}$ nor in $\mathbb{S}$.
\end{proof}

The proposition above guarantees that solving the $\sigma_8$ disambiguation or finding an 8-maximum solution of a graph $G_8^{+\nu}(\mathcal{X},\mathcal{Y})$ is equivalent to solving the $\sigma_8$ double distance of the corresponding linear genomes $\mathbb{S}$ and $\mathbb{D}$.
Here we need to compute the 0-paths of the graph, therefore the relation between the solution of the $\satvar$-SAT instance $(\mathcal{X},\mathcal{Y})$ and an 8-maximum solution of $G_8^{+\nu}(\mathcal{X},\mathcal{Y})$ is as follows.

\begin{lemma}\label{lmm:iff-lin}
For any 8-maximum solution $\tau$ of graph $G_8^{+\nu}(\mathcal{X},\mathcal{Y})$, it holds that:
(i) $\textup{s}_8(\tau) \leq |\mathcal{X}|+|\mathcal{Y}|+\|\mathcal{Y}\| + \frac{\nu_8}{2}$;
and (ii) formula $\mathcal{Y}$ has a satisfying assignment if and only if $\textup{s}_8(\tau) = |\mathcal{X}|+|\mathcal{Y}|+\|\mathcal{Y}\| + \frac{\nu_8}{2}$.
\end{lemma}
\begin{proof}
Note that all $\nu_8$ isolated vertices are 0-paths, which are part of the score of any solution. The proof is then a direct consequence of the proof of Lemma~\ref{lmm:iff-circ}, but here each of the $\nu_8$ 0-paths of the graph contributes with $\frac{1}{2}$ to the score of any solution.
\end{proof}

\subsection{NP-completeness of \boldmath$\sigma_8$ disambiguation}

The NP-hardness of $\sigma_8$ disambiguations  of circular and linear genomes follows directly from Lemmas~\ref{lmm:iff-circ} and~\ref{lmm:iff-lin}.
Moreover, the problem is in NP since a solution for the $\sigma_8$ disambiguation can be easily verified in polynomial time because
a certificate is a collection of cycles, of which one has to count how many have length at most~8, and for linear genomes also of paths, of which one has to count how many have even length at most~6.  Therefore:

\begin{theorem}\label{thm:NPc}
Given a pair of $\cogsd$-cognate genomes $\mathbb{S}$ and $\mathbb{D}$, both being linear or both being circular, computing the $\sigma_8$ double distance or the $\sigma_8$ disambiguation of $\mathbb{S}$ and $\mathbb{D}$ is NP-complete.
\end{theorem}

Circular and
linear genomes are special cases of \emph{mixed genomes}, composed of a combination
of linear and circular chromosomes. Consequently, the theorem above can be generalized to mixed genomes:

\begin{corollary}
Computing the $\sigma_8$ double distance or the $\sigma_8$ disambiguation of a pair of $\cogsd$-cognate mixed genomes $\mathbb{S}$ and $\mathbb{D}$ is NP-complete.
\end{corollary}

\subsection{NP-completeness of \boldmath$\sigma_{\geq10}$ disambiguation}



The graph $G_8(\mathcal{X},\mathcal{Y})$ has alternating cycles whose lengths are greater than 8 and some of these cycles can invalidate our reduction for $k>8$ if we do not take particular care of them. Critical cases involve the gadget for a 3-clause, as shown in Fig.~\ref{fig:12-cycles}. 

\begin{figure}[ht]

    \begin{center}
     \begin{minipage}{3cm}
     \begin{center}

     \textbf{(a)}

     \vspace{4.5mm}
     
    \scalebox{0.32}{
\begin{tikzpicture}[scale=0.7]

    \node[intedgecolor] at (2,2) {\huge $\Theta_1$};
    \node[intedgecolor] at (6,2) {\huge $\Theta_2$};
    \node[intedgecolor] at (10,2) {\huge $\Theta_3$};

    \vcl{(0,3)}{q1v1}{}
	\vcl{(1,4)}{q1v2}{}
    \vcl{(0,5)}{q1v3}{}
	\vcl{(-1,4)}{q1v4}{}
    
    \vcl{(4,3)}{q2v1}{}
	\vcl{(5,4)}{q2v2}{}
    \vcl{(4,5)}{q2v3}{}
	\vcl{(3,4)}{q2v4}{}
    
    \vcl{(8,3)}{q3v1}{}
	\vcl{(9,4)}{q3v2}{}
    \vcl{(8,5)}{q3v3}{}
	\vcl{(7,4)}{q3v4}{}

    \vcl{(0,-1)}{q4v1}{}
	\vcl{(1,0)}{q4v2}{}
    \vcl{(0,1)}{q4v3}{}
	\vcl{(-1,0)}{q4v4}{}
    
    \vcl{(4,-1)}{q5v1}{}
	\vcl{(5,0)}{q5v2}{}
    \vcl{(4,1)}{q5v3}{}
	\vcl{(3,0)}{q5v4}{}
    
    \vcl{(8,-1)}{q6v1}{}
	\vcl{(9,0)}{q6v2}{}
    \vcl{(8,1)}{q6v3}{}
	\vcl{(7,0)}{q6v4}{}

    \ed{q1v1}{q1v4}{cyan}
    \ed{q1v3}{q1v2}{cyan}
    \draw[intedgecolor,dotted,ultra thick] (q1v3) -- (0,6) node[anchor=west]{\LARGE $\ope{W_1}$};
            
    \ed{q2v1}{q2v2}{cyan}
    \ed{q2v3}{q2v4}{cyan}
    \draw[intedgecolor,dotted,ultra thick] (q2v3) -- (4,6) node[anchor=west]{\LARGE $\ope{W_1}$};

    \ed{q3v1}{q3v4}{cyan}
    \ed{q3v3}{q3v2}{cyan}
    \draw[intedgecolor,ultra thick] (q3v3) -- (8,6) node[anchor=west]{\LARGE $\ope{W_3}$};

    \ed{q4v1}{q4v4}{cyan}
    \ed{q4v3}{q4v2}{cyan}
    \draw[intedgecolor,ultra thick] (q4v1) -- (0,-2) node[anchor=west]{\LARGE $\ope{W_3}$};

    \ed{q5v1}{q5v2}{cyan}
    \ed{q5v3}{q5v4}{cyan}
    \draw[intedgecolor,dashed,ultra thick] (q5v1) -- (4,-2) node[anchor=west]{\LARGE $\ope{W_2}$};

    \ed{q6v1}{q6v4}{cyan}
    \ed{q6v3}{q6v2}{cyan}
    \draw[intedgecolor,dashed,ultra thick] (q6v1) -- (8,-2) node[anchor=west]{\LARGE $\ope{W_2}$};
    
    \ed{q1v2}{q2v4}{intedgecolor}
    \ed{q1v1}{q4v3}{black}
    \ed{q2v2}{q3v4}{black}
    \ed{q2v1}{q5v3}{black}
    \ed{q3v1}{q6v3}{black}
    \ed{q4v2}{q5v4}{black}
    \ed{q5v2}{q6v4}{intedgecolor}

    \draw[intedgecolor,ultra thick] (q3v2) -- (11.5,4) -- (11.5,-3) -- (-2,-3) -- (-2, 0) -- (q4v4);

    \draw[black,ultra thick] (q1v4) -- (-2, 4) -- (-2, 7) -- (11,7) -- (11,0) -- (q6v2);

\end{tikzpicture}
}

    \vspace{8.5mm}
    
     \end{center}
     \end{minipage}
    \hspace{4mm}
    \begin{minipage}{4.5cm}
     \begin{center}
     \textbf{(b)}

     \medskip
     
    \scalebox{0.32}{
\begin{tikzpicture}[scale=0.7]

    \node[intedgecolor] at (2,2) {\huge $\Theta_1$};
    \node[intedgecolor] at (6,2) {\huge $\Theta_2$};
    \node[intedgecolor] at (10,2) {\huge $\Theta_3$};

    \vcl{(0,3)}{q1v1}{}
	\vcl{(1,4)}{q1v2}{}
    \vcl{(0,5)}{q1v3}{}
	\vcl{(-1,4)}{q1v4}{}
    
    \vcl{(4,3)}{q2v1}{}
	\vcl{(5,4)}{q2v2}{}
    \vcl{(4,5)}{q2v3}{}
	\vcl{(3,4)}{q2v4}{}
    
    \vcl{(8,3)}{q3v1}{}
	\vcl{(9,4)}{q3v2}{}
    \vcl{(8,5)}{q3v3}{}
	\vcl{(7,4)}{q3v4}{}

    \vcl{(0,-1)}{q4v1}{}
	\vcl{(1,0)}{q4v2}{}
    \vcl{(0,1)}{q4v3}{}
	\vcl{(-1,0)}{q4v4}{}
    
    \vcl{(4,-1)}{q5v1}{}
	\vcl{(5,0)}{q5v2}{}
    \vcl{(4,1)}{q5v3}{}
	\vcl{(3,0)}{q5v4}{}
    
    \vcl{(8,-1)}{q6v1}{}
	\vcl{(9,0)}{q6v2}{}
    \vcl{(8,1)}{q6v3}{}
	\vcl{(7,0)}{q6v4}{}

    \ed{q1v1}{q1v4}{cyan}
    \ed{q1v3}{q1v2}{cyan}
    \draw[intedgecolor,dotted,ultra thick] (q1v3) -- (0,6) node[anchor=west]{\LARGE $\ope{W_1}$};
            
    \ed{q2v1}{q2v2}{cyan}
    \ed{q2v3}{q2v4}{cyan}
    \draw[intedgecolor,dotted,ultra thick] (q2v3) -- (4,6) node[anchor=west]{\LARGE $\ope{W_1}$};

    \ed{q3v1}{q3v2}{cyan}
    \ed{q3v3}{q3v4}{cyan}
    \draw[intedgecolor,ultra thick] (q3v3) -- (8,6) node[anchor=west]{\LARGE $\ope{W_3}$};

    \ed{q4v1}{q4v2}{cyan}
    \ed{q4v3}{q4v4}{cyan}
    \draw[intedgecolor,ultra thick] (q4v1) -- (0,-2) node[anchor=west]{\LARGE $\ope{W_3}$};

    \ed{q5v1}{q5v2}{cyan}
    \ed{q5v3}{q5v4}{cyan}
    \draw[intedgecolor,dashed,ultra thick] (q5v1) -- (4,-2) node[anchor=west]{\LARGE $\ope{W_2}$};

    \ed{q6v1}{q6v4}{cyan}
    \ed{q6v3}{q6v2}{cyan}
    \draw[intedgecolor,dashed,ultra thick] (q6v1) -- (8,-2) node[anchor=west]{\LARGE $\ope{W_2}$};
    
    \ed{q1v2}{q2v4}{intedgecolor}
    \ed{q1v1}{q4v3}{intedgecolor}
    \ed{q2v2}{q3v4}{black}
    \ed{q2v1}{q5v3}{black}
    \ed{q3v1}{q6v3}{intedgecolor}
    \ed{q4v2}{q5v4}{black}
    \ed{q5v2}{q6v4}{intedgecolor}

    \draw[intedgecolor,ultra thick] (q3v2) -- (11.5,4) -- (11.5,-3) -- (-2,-3) -- (-2, 0) -- (q4v4);

    \draw[intedgecolor,ultra thick] (q1v4) -- (-2, 4) -- (-2, 7) -- (11,7) -- (11,0) -- (q6v2);
    
\begin{scope}[shift={(-5,-8)}]
   \vcl{(1,0)}{wq1v1}{}
	\tvcl{(2.4,0)}{wq1v2}{}{black}
    \vcl{(2.4,1.4)}{wq1v3}{}
	\vcl{(1,1.4)}{wq1v4}{}

    \vcl{(1,3)}{wq2v1}{}
	\vcl{(2.4,3)}{wq2v2}{}
    \tvcl{(2.4,4.4)}{wq2v3}{}{black}
	\vcl{(1,4.4)}{wq2v4}{}
   
    \ed{wq1v1}{wq1v2}{cyan} 
    \ed{wq1v3}{wq1v4}{cyan}
    \draw[intedgecolor,ultra thick] (wq1v1) -- (0,0) node[anchor=east]{\LARGE $\ope{X}$};
            
    \ed{wq2v1}{wq2v2}{cyan}
    \ed{wq2v3}{wq2v4}{cyan}
    \draw[intedgecolor,ultra thick] (wq2v4) -- (0,4.4) node[anchor=east]{\LARGE $\ope{X}$};
            
    \ed{wq1v4}{wq2v1}{black}
\end{scope}

\draw[black,ultra thick] (q4v1) -- (0,-5) -- (wq2v2);

\draw[black,ultra thick] (q3v3) -- (8,8) -- (13,8) -- (13,-6.6) -- (wq1v3);

\end{tikzpicture}
}
    \end{center}
     \end{minipage}
     
    \end{center}

\caption{
The gadget for a 3-clause allows critical 12-cycles:
  \figcbf{(a)} A 12-cycle that competes with $\Theta_1$, $\Theta_2$ and $\Theta_3$, but is compatible with outer connections $\ope{W_1}$, $\ope{W_2}$ and $\ope{W_3}$.
  \figcbf{(b)} A 12-cycle (including  outer connections $\ope{W_3}$) that competes with $\Theta_1$ and $\Theta_2$, but is compatible with $\Theta_3$ and outer connections $\ope{W_1}$ and $\ope{W_2}$.
  }\label{fig:12-cycles}
\end{figure}

Due to this fact, we will proceed with a polynomial adaptation of the graph that safely allows us to reduce $\satvar$-SAT to $\sigma_k$ disambiguation for any fixed $k \geq 10$. This adaptation consists of extending some edges of the graph by adding extra squares plugged to open flowers, as depicted in Fig.~\ref{fig:edge-extension}. The idea is to modify the graph such that the lengths of all 8-cycles are increased to $k$. As a consequence, all other alternating cycles will certainly have at least length $k+2$, meaning that no new candidate cycle can appear and the correspondence to the $\satvar$-SAT solution is preserved. 
Fig.~\ref{fig:new-gadgets} shows the edges that need to be extended to achieve this goal.
Given $k \geq 10$, let $G_{k}(\mathcal{X},\mathcal{Y})$ be the graph obtained from
$G_8(\mathcal{X},\mathcal{Y})$ by
adapting its gadgets as described in Fig.~\ref{fig:new-gadgets}.

\begin{figure}[ht]
    \centering
    \scalebox{0.4}{
\begin{minipage}{1cm}
\centering
\begin{tikzpicture}[scale=0.8]

    \bigbigvcl{(0,0)}{start}{$v$}
    \bigbigvcl{(0,4.5)}{end}{$u$}
    \extpath{start}{end}{$0$}
\end{tikzpicture}
\end{minipage}
\begin{minipage}{1cm}
\centering
{\huge$\Rightarrow$}
\end{minipage}
\begin{minipage}{1cm}
\centering
\begin{tikzpicture}[scale=0.8]    
    \bigbigvcl{(0,0)}{v0}{$v$}
    \bigbigvcl{(0,4.5)}{vT}{$u$}
    \ed{v0}{vT}{black}
\end{tikzpicture}
\end{minipage}
\hspace{2cm}
\begin{minipage}{1cm}
\centering
\begin{tikzpicture}[scale=0.8]

    \bigbigvcl{(0,0)}{start}{$v$}
    \bigbigvcl{(0,4.5)}{end}{$u$}
    \extpath{start}{end}{$1$}
\end{tikzpicture}
\end{minipage}
\begin{minipage}{1cm}
\centering
{\huge$\Rightarrow$}
\end{minipage}
\begin{minipage}{2cm}
\centering
\begin{tikzpicture}[scale=0.8]    
    \bigbigvcl{(0,0)}{v0}{$v$}
    
    \vcl{(0,1.5)}{q1v1}{}
	\tvcl{(1.5,1.5)}{q1v2}{}{black}
    \tvcl{(1.5,3)}{q1v3}{}{black}
	\vcl{(0,3)}{q1v4}{}
    \bigbigvcl{(0,4.5)}{vT}{$u$}

    \node[orange] at (0.75,2.25) {\Large $1$};
   
    \dted{q1v1}{q1v2}{orange} 
    \ed{q1v1}{q1v4}{orange}
    \ed{q1v3}{q1v2}{orange}
    \dted{q1v3}{q1v4}{orange}

    \ed{v0}{q1v1}{black}
    \ed{vT}{q1v4}{black}
\end{tikzpicture}
\end{minipage}
\hspace{2cm}
\begin{minipage}{10cm}
\centering

\begin{tikzpicture}[scale=0.8]

    \bigbigvcl{(0,0)}{start}{$u$}
    \bigbigvcl{(4.5,0)}{end}{$v$}
    \extpath{start}{end}{$\ell$}
\end{tikzpicture}

\bigskip

{\huge$\Downarrow$}

\bigskip

\begin{tikzpicture}[scale=0.8]    
    \bigbigvcl{(0,1.5)}{v0}{$u$}
    
    \tvcl{(1.5,0)}{q1v1}{}{black}
	\tvcl{(3,0)}{q1v2}{}{black}
    \vcl{(3,1.5)}{q1v3}{}
	\vcl{(1.5,1.5)}{q1v4}{}

    \node[orange] at (2.25,0.75) {\Large $1$};

    \tvcl{(4.5,0)}{q2v1}{}{black}
	\tvcl{(6,0)}{q2v2}{}{black}
    \vcl{(6,1.5)}{q2v3}{}
	\vcl{(4.5,1.5)}{q2v4}{}

    \node[orange] at (5.25,0.75) {\Large $2$};

    \node[orange] at (7.25,0.75) {\LARGE {\boldmath $\ldots$}};

    \tvcl{(8.5,0)}{qLv1}{}{black}
	\tvcl{(10,0)}{qLv2}{}{black}
    \vcl{(10,1.5)}{qLv3}{}
	\vcl{(8.5,1.5)}{qLv4}{}

     \node[orange] at (9.25,0.75) {\Large $\ell$};

    \bigbigvcl{(11.5,1.5)}{vT}{$v$}

    \ed{q1v1}{q1v2}{orange} 
    \dted{q1v1}{q1v4}{orange}
    \dted{q1v3}{q1v2}{orange}
    \ed{q1v3}{q1v4}{orange}

    \ed{q2v1}{q2v2}{orange} 
    \dted{q2v1}{q2v4}{orange}
    \dted{q2v3}{q2v2}{orange}
    \ed{q2v3}{q2v4}{orange}

    \ed{qLv1}{qLv2}{orange} 
    \dted{qLv1}{qLv4}{orange}
    \dted{qLv3}{qLv2}{orange}
    \ed{qLv3}{qLv4}{orange}

    \ed{q1v3}{q2v4}{black}
    \dted{q2v3}{qLv4}{black}
    
    \ed{v0}{q1v4}{black}
    \ed{vT}{qLv3}{black}
    
\end{tikzpicture}

\end{minipage}
}
    
\caption{Extending a $\check{\mathbb{D}}$-edge $uv$ of the graph to longer alternating paths by adding $0$, $1$ or $\ell$ squares. Each square  is plugged to an open flower and increases by 2 the length of any cycle including $uv$.
}\label{fig:edge-extension}
\end{figure}

\begin{figure*}[ht]

       \begin{center}

\begin{minipage}{2cm}
    \begin{center}
    ~
    \end{center}
\end{minipage}
\hspace{2cm}
\begin{minipage}{3.5cm}
    \begin{center}
    \textbf{(a) $\mathbi{X}$-$\mathtt{TTF}$}

    \smallskip

    \scalebox{0.35}{\begin{tikzpicture}[scale=0.8]

    \node at (2,2) {\LARGE $\Theta_{\mathtt{F}}$};
    \node at (6,2) {\LARGE $\Theta_{\mathtt{T}}$};
    
    \vcl{(0,3)}{q1v1}{}
	\vcl{(1,4)}{q1v2}{}
    \vcl{(0,5)}{q1v3}{}
	\tvcl{(-1,4)}{q1v4}{}{black}

    \vcl{(4,3)}{q2v1}{}
	\vcl{(5,4)}{q2v2}{}
    \vcl{(4,5)}{q2v3}{}
	\vcl{(3,4)}{q2v4}{}

    \vcl{(8,3)}{q3v1}{}
	\vcl{(9,4)}{q3v2}{}
    \tvcl{(8,5)}{q3v3}{}{black}
	\vcl{(7,4)}{q3v4}{}
    
    \vcl{(0,-1)}{q4v1}{}
	\vcl{(1,0)}{q4v2}{}
    \vcl{(0,1)}{q4v3}{}
	\tvcl{(-1,0)}{q4v4}{}{black}
    
    \vcl{(4,-1)}{q5v1}{}
	\vcl{(5,0)}{q5v2}{}
    \vcl{(4,1)}{q5v3}{}
	\vcl{(3,0)}{q5v4}{}
    
    \tvcl{(8,-1)}{q6v1}{}{black}
	\vcl{(9,0)}{q6v2}{}
    \vcl{(8,1)}{q6v3}{}
	\vcl{(7,0)}{q6v4}{}

    \ed{q1v1}{q1v2}{orange}
    \ed{q1v1}{q1v4}{orange}
    \ed{q1v3}{q1v2}{orange}
    \ed{q1v3}{q1v4}{orange}
    \draw[intedgecolor,ultra thick] (q1v3) -- (0,6) node[anchor=west]{\LARGE $\ope{\mathtt{T}_1}$};
            
    \ed{q2v1}{q2v2}{orange}
    \ed{q2v1}{q2v4}{orange}
    \ed{q2v3}{q2v2}{orange}
    \ed{q2v3}{q2v4}{orange}
    \draw[intedgecolor,ultra thick] (q2v3) -- (4,6) node[anchor=west]{\LARGE $\ope{\mathtt{T}_1}$};

    \ed{q3v1}{q3v2}{orange}
    \ed{q3v1}{q3v4}{orange}
    \ed{q3v3}{q3v2}{orange}
    \ed{q3v3}{q3v4}{orange}
    \draw[intedgecolor,ultra thick] (q3v2) -- (10,4) node[anchor=west]{\LARGE $\ope{\mathtt{F}}$};

    \ed{q4v1}{q4v2}{orange}
    \ed{q4v1}{q4v4}{orange}
    \ed{q4v3}{q4v2}{orange}
    \ed{q4v3}{q4v4}{orange}
    \draw[intedgecolor,ultra thick] (q4v1) -- (0,-2) node[anchor=west]{\LARGE $\ope{\mathtt{T}_2}$};

    \ed{q5v1}{q5v2}{orange}
    \ed{q5v1}{q5v4}{orange}
    \ed{q5v3}{q5v2}{orange}
    \ed{q5v3}{q5v4}{orange}
     \draw[intedgecolor,ultra thick] (q5v1) -- (4,-2) node[anchor=west]{\LARGE $\ope{\mathtt{T}_2}$};

    \ed{q6v1}{q6v2}{orange}
    \ed{q6v1}{q6v4}{orange}
    \ed{q6v3}{q6v2}{orange}
    \ed{q6v3}{q6v4}{orange}
    \draw[intedgecolor,ultra thick] (q6v2) -- (10,0) node[anchor=west]{\LARGE $\ope{\mathtt{F}}$};
    
    \ed{q1v2}{q2v4}{intedgecolor}
    \ed{q1v1}{q4v3}{intedgecolor}
    \ed{q2v2}{q3v4}{intedgecolor}
    \extpath{q2v1}{q5v3}{$\ell$}
    \ed{q3v1}{q6v3}{intedgecolor}
    \ed{q4v2}{q5v4}{intedgecolor}
    \ed{q5v2}{q6v4}{intedgecolor}

\end{tikzpicture}
}
    \end{center}
 \end{minipage}
 \hspace{1.5cm}
 \begin{minipage}{4cm}

    \begin{center}
\textbf{(b) $\mathbi{X}$-$\mathtt{TF}$}

    \smallskip

    \scalebox{0.35}{\begin{tikzpicture}[scale=0.8]

    \node at (2,2) {\LARGE $\Theta_{\mathtt{F}}$};
    \node at (6,2) {\LARGE $\Theta_{\mathtt{T}}$};
    
    \vcl{(0,3)}{q1v1}{}
	\vcl{(1,4)}{q1v2}{}
    \vcl{(0,5)}{q1v3}{}
	\tvcl{(-1,4)}{q1v4}{}{black}

    \vcl{(4,3)}{q2v1}{}
	\vcl{(5,4)}{q2v2}{}
    \vcl{(4,5)}{q2v3}{}
	\vcl{(3,4)}{q2v4}{}

    \vcl{(8,3)}{q3v1}{}
	\vcl{(9,4)}{q3v2}{}
    \tvcl{(8,5)}{q3v3}{}{black}
	\vcl{(7,4)}{q3v4}{}

    \tvcl{(0,-1)}{q4v1}{}{black}
	\vcl{(1,0)}{q4v2}{}
    \vcl{(0,1)}{q4v3}{}
	\tvcl{(-1,0)}{q4v4}{}{black}

    \tvcl{(4,-1)}{q5v1}{}{black}
	\vcl{(5,0)}{q5v2}{}
    \vcl{(4,1)}{q5v3}{}
	\vcl{(3,0)}{q5v4}{}

    \tvcl{(8,-1)}{q6v1}{}{black}
	\vcl{(9,0)}{q6v2}{}
    \vcl{(8,1)}{q6v3}{}
	\vcl{(7,0)}{q6v4}{}

    \ed{q1v1}{q1v2}{orange}
    \ed{q1v1}{q1v4}{orange}
    \ed{q1v3}{q1v2}{orange}
    \ed{q1v3}{q1v4}{orange}
    \draw[intedgecolor,ultra thick] (q1v3) -- (0,6) node[anchor=west]{\LARGE $\ope{\mathtt{T}}$};
            
    \ed{q2v1}{q2v2}{orange}
    \ed{q2v1}{q2v4}{orange}
    \ed{q2v3}{q2v2}{orange}
    \ed{q2v3}{q2v4}{orange}
    \draw[intedgecolor,ultra thick] (q2v3) -- (4,6) node[anchor=west]{\LARGE $\ope{\mathtt{T}}$};

    \ed{q3v1}{q3v2}{orange}
    \ed{q3v1}{q3v4}{orange}
    \ed{q3v3}{q3v2}{orange}
    \ed{q3v3}{q3v4}{orange}
    \draw[intedgecolor,ultra thick] (q3v2) -- (10,4) node[anchor=west]{\LARGE $\ope{\mathtt{F}}$};

    \ed{q4v1}{q4v2}{orange}
    \ed{q4v1}{q4v4}{orange}
    \ed{q4v3}{q4v2}{orange}
    \ed{q4v3}{q4v4}{orange}

    \ed{q5v1}{q5v2}{orange}
    \ed{q5v1}{q5v4}{orange}
    \ed{q5v3}{q5v2}{orange}
    \ed{q5v3}{q5v4}{orange}

    \ed{q6v1}{q6v2}{orange}
    \ed{q6v1}{q6v4}{orange}
    \ed{q6v3}{q6v2}{orange}
    \ed{q6v3}{q6v4}{orange}
    \draw[intedgecolor,ultra thick] (q6v2) -- (10,0) node[anchor=west]{\LARGE $\ope{\mathtt{F}}$};
    
    \ed{q1v2}{q2v4}{intedgecolor}
    \ed{q1v1}{q4v3}{intedgecolor}
    \ed{q2v2}{q3v4}{intedgecolor}
    \extpath{q2v1}{q5v3}{$\ell$}
    \ed{q3v1}{q6v3}{intedgecolor}
    \ed{q4v2}{q5v4}{intedgecolor}
    \ed{q5v2}{q6v4}{intedgecolor}

\end{tikzpicture}
}
    
   \vspace{3mm}
  \end{center}
 \end{minipage}
 
        \medskip

\begin{minipage}{2cm}
    \begin{center}
      \textbf{(c) $\mathbi{W}$}

      \vspace{5mm}
       
      \scalebox{0.35}{
\begin{tikzpicture}[scale=0.8]

    \vcl{(1,0)}{q1v1}{}
	\tvcl{(2.4,0)}{q1v2}{}{black}
    \vcl{(2.4,1.4)}{q1v3}{}
	\vcl{(1,1.4)}{q1v4}{}

    \vcl{(1,4)}{q2v1}{}
	\vcl{(2.4,4)}{q2v2}{}
    \tvcl{(2.4,5.4)}{q2v3}{}{black}
	\vcl{(1,5.4)}{q2v4}{}
   
    \dted{q1v1}{q1v2}{orange} 
    \ed{q1v1}{q1v4}{orange}
    \ed{q1v3}{q1v2}{orange}
    \dted{q1v3}{q1v4}{orange}
    \draw[intedgecolor,ultra thick] (q1v1) -- (0,0) node[anchor=east]{\LARGE $\ope{X}$};
    \draw[intedgecolor,ultra thick] (q1v3) -- (3.4,1.4) node[anchor=west]{\LARGE $\ope{Y}$};
            
    \dted{q2v1}{q2v2}{orange}
    \ed{q2v1}{q2v4}{orange}
    \ed{q2v3}{q2v2}{orange}
    \dted{q2v3}{q2v4}{orange}
    \draw[intedgecolor,ultra thick] (q2v4) -- (0,5.4) node[anchor=east]{\LARGE $\ope{X}$};
    \draw[intedgecolor,ultra thick] (q2v2) -- (3.4,4) node[anchor=west]{\LARGE $\ope{Y}$};
            
    \extpath{q1v4}{q2v1}{$\ell$}

\end{tikzpicture}
}

      \vspace{5mm}
    \end{center}  
\end{minipage}
\hspace{2cm}
\begin{minipage}{3.5cm}
    \begin{center}
    \textbf{(d) $2$-\!$\mathbi{Y}$}

    \vspace{7mm}
    
    \hspace{-1.2cm}\scalebox{0.35}{
\begin{tikzpicture}[scale=0.8]

    \node at (2,2) {\LARGE $\Theta_1$};
    \node at (6,2) {\LARGE $\Theta_2$};

    \vcl{(0,3)}{q1v1}{}
	\vcl{(1,4)}{q1v2}{}
    \tvcl{(0,5)}{q1v3}{}{black}
	\vcl{(-1,4)}{q1v4}{}

    \vcl{(4,3)}{q2v1}{}
	\vcl{(5,4)}{q2v2}{}
    \tvcl{(4,5)}{q2v3}{}{black}
	\vcl{(3,4)}{q2v4}{}
    
    \vcl{(8,3)}{q3v1}{}
	\tvcl{(9,4)}{q3v2}{}{black}
    \tvcl{(8,5)}{q3v3}{}{black}
	\vcl{(7,4)}{q3v4}{}

    \tvcl{(0,-1)}{q4v1}{}{black}
	\vcl{(1,0)}{q4v2}{}
    \vcl{(0,1)}{q4v3}{}
	\vcl{(-1,0)}{q4v4}{}

    \vcl{(4,-1)}{q5v1}{}
	\vcl{(5,0)}{q5v2}{}
    \vcl{(4,1)}{q5v3}{}
	\vcl{(3,0)}{q5v4}{}

    \vcl{(8,-1)}{q6v1}{}
	\tvcl{(9,0)}{q6v2}{}{black}
    \vcl{(8,1)}{q6v3}{}
	\vcl{(7,0)}{q6v4}{}

    \ed{q1v1}{q1v2}{orange}
    \ed{q1v1}{q1v4}{orange}
    \ed{q1v3}{q1v2}{orange}
    \ed{q1v3}{q1v4}{orange}
    \draw[intedgecolor,ultra thick] (q1v4) -- (-2,4) node[anchor=east]{\LARGE $\ope{W_1}$};
            
    \ed{q2v1}{q2v2}{orange}
    \ed{q2v1}{q2v4}{orange}
    \ed{q2v3}{q2v2}{orange}
    \ed{q2v3}{q2v4}{orange}

    \ed{q3v1}{q3v2}{orange}
    \ed{q3v1}{q3v4}{orange}
    \ed{q3v3}{q3v2}{orange}
    \ed{q3v3}{q3v4}{orange}

    \ed{q4v1}{q4v2}{orange}
    \ed{q4v1}{q4v4}{orange}
    \ed{q4v3}{q4v2}{orange}
    \ed{q4v3}{q4v4}{orange}
    \draw[intedgecolor,ultra thick] (q4v4) -- (-2,0) node[anchor=east]{\LARGE $\ope{W_1}$};

    \ed{q5v1}{q5v2}{orange}
    \ed{q5v1}{q5v4}{orange}
    \ed{q5v3}{q5v2}{orange}
    \ed{q5v3}{q5v4}{orange}
    \draw[intedgecolor,ultra thick] (q5v1) -- (4,-2) node[anchor=west]{\LARGE $\ope{W_2}$};

    \ed{q6v1}{q6v2}{orange}
    \ed{q6v1}{q6v4}{orange}
    \ed{q6v3}{q6v2}{orange}
    \ed{q6v3}{q6v4}{orange}
    \draw[intedgecolor,ultra thick] (q6v1) -- (8,-2) node[anchor=west]{\LARGE $\ope{W_2}$};
    
    \ed{q1v2}{q2v4}{intedgecolor}
    \ed{q1v1}{q4v3}{intedgecolor}
    \ed{q2v2}{q3v4}{intedgecolor}
    \extpath{q2v1}{q5v3}{$\ell$}
    \ed{q3v1}{q6v3}{intedgecolor}
    \ed{q4v2}{q5v4}{intedgecolor}
    \ed{q5v2}{q6v4}{intedgecolor}

\end{tikzpicture}
}
    
    \vspace{2mm}
    \end{center}
 \end{minipage}
 \hspace{1.5cm}
 \begin{minipage}{4cm}
    \begin{center}
    \textbf{(e) $3$-\!$\mathbi{Y}$}\\

     \smallskip
     
    \scalebox{0.35}{
\begin{tikzpicture}[scale=0.8]

    \node at (2,2) {\LARGE $\Theta_1$};
    \node at (6,2) {\LARGE $\Theta_2$};
    \node at (10,2) {\LARGE $\Theta_3$};

    \vcl{(0,3)}{q1v1}{}
	\vcl{(1,4)}{q1v2}{}
    \vcl{(0,5)}{q1v3}{}
	\vcl{(-1,4)}{q1v4}{}
    
    \vcl{(4,3)}{q2v1}{}
	\vcl{(5,4)}{q2v2}{}
    \vcl{(4,5)}{q2v3}{}
	\vcl{(3,4)}{q2v4}{}
    
    \vcl{(8,3)}{q3v1}{}
	\vcl{(9,4)}{q3v2}{}
    \vcl{(8,5)}{q3v3}{}
	\vcl{(7,4)}{q3v4}{}

    \vcl{(0,-1)}{q4v1}{}
	\vcl{(1,0)}{q4v2}{}
    \vcl{(0,1)}{q4v3}{}
	\vcl{(-1,0)}{q4v4}{}
    
    \vcl{(4,-1)}{q5v1}{}
	\vcl{(5,0)}{q5v2}{}
    \vcl{(4,1)}{q5v3}{}
	\vcl{(3,0)}{q5v4}{}
    
    \vcl{(8,-1)}{q6v1}{}
	\vcl{(9,0)}{q6v2}{}
    \vcl{(8,1)}{q6v3}{}
	\vcl{(7,0)}{q6v4}{}

    \ed{q1v1}{q1v2}{orange}
    \ed{q1v1}{q1v4}{orange}
    \ed{q1v3}{q1v2}{orange}
    \ed{q1v3}{q1v4}{orange}
    \draw[intedgecolor,ultra thick] (q1v3) -- (0,6) node[anchor=west]{\LARGE $\ope{W_1}$};
            
    \ed{q2v1}{q2v2}{orange}
    \ed{q2v1}{q2v4}{orange}
    \ed{q2v3}{q2v2}{orange}
    \ed{q2v3}{q2v4}{orange}
    \draw[intedgecolor,ultra thick] (q2v3) -- (4,6) node[anchor=west]{\LARGE $\ope{W_1}$};

    \ed{q3v1}{q3v2}{orange}
    \ed{q3v1}{q3v4}{orange}
    \ed{q3v3}{q3v2}{orange}
    \ed{q3v3}{q3v4}{orange}
    \draw[intedgecolor,ultra thick] (q3v3) -- (8,6) node[anchor=west]{\LARGE $\ope{W_3}$};

    \ed{q4v1}{q4v2}{orange}
    \ed{q4v1}{q4v4}{orange}
    \ed{q4v3}{q4v2}{orange}
    \ed{q4v3}{q4v4}{orange}
    \draw[intedgecolor,ultra thick] (q4v1) -- (0,-2) node[anchor=west]{\LARGE $\ope{W_3}$};

    \ed{q5v1}{q5v2}{orange}
    \ed{q5v1}{q5v4}{orange}
    \ed{q5v3}{q5v2}{orange}
    \ed{q5v3}{q5v4}{orange}
    \draw[intedgecolor,ultra thick] (q5v1) -- (4,-2) node[anchor=west]{\LARGE $\ope{W_2}$};

    \ed{q6v1}{q6v2}{orange}
    \ed{q6v1}{q6v4}{orange}
    \ed{q6v3}{q6v2}{orange}
    \ed{q6v3}{q6v4}{orange}
    \draw[intedgecolor,ultra thick] (q6v1) -- (8,-2) node[anchor=west]{\LARGE $\ope{W_2}$};
    
    \ed{q1v2}{q2v4}{intedgecolor}
    \ed{q1v1}{q4v3}{intedgecolor}
    \ed{q2v2}{q3v4}{intedgecolor}
    \extpath{q2v1}{q5v3}{$\ell$}
    \ed{q3v1}{q6v3}{intedgecolor}
    \ed{q4v2}{q5v4}{intedgecolor}
    \ed{q5v2}{q6v4}{intedgecolor}

    \draw[intedgecolor,ultra thick] (q3v2) -- (11.5,4) -- (11.5,-3) -- (-2,-3) -- (-2, 0) -- (q4v4);

    \draw[path,dotted,line width=3pt] (q1v4) -- (-2, 4) -- (-2, 7) -- (11,7) node[draw=path,thick,solid,fill=white,font=\Large,midway] {$\ell$} -- (11,0) -- (q6v2);

\end{tikzpicture}
}
    
    \end{center}
 \end{minipage}

   \end{center}
\caption{Adaptation of the gadgets for $k\geq 10$, with $\ell=\frac{k-8}{2}$ and $p=\frac{k}{2}+1$. In all gadgets the $5$-flowers must be replaced by $p$-flowers, and in all edge extensions each added square is plugged to an extra $p$-flower. In both \figcbf{(a)} and \figcbf{(b)} an extension occurs in the edge shared by $\Theta_{\mathtt{F}}$ and $\Theta_{\mathtt{T}}$, increasing their lengths from 8 to $k$. In \figcbf{(c)} the extension occurs in the edge connecting the two squares. This guarantees that all cycles obtained by joining outer connections of the gadgets will have the lengths increased from 8 to $k$. In both \figcbf{(d)} and \figcbf{(e)} the extension occurs in the edge shared by $\Theta_1$ and $\Theta_2$, increasing their lengths from 8 to $k$. Finally, in \figcbf{(e)} an extra entension occurs in one edge belonging only to $\Theta_3$, increasing its length from 8 to $k$. Note that this last extension does not affect the outer connections of the gadget. 
In summary, except for the gadget of a 3-clause, that requires the extension of two edges, we need to extend one edge per gadget. Each of these extensions consists of adding $\ell$ extra squares, each one plugged to a $p$-flower. In total, the number of edges that must be extended is $m=|\mathcal{X}|+\|\mathcal{Y}\|+|\mathcal{Y}|+|\mathcal{Y}_3|$, where $\mathcal{Y}_3$ is the subset of 3-clauses in $\mathcal{Y}$. The number of extra squares and extra $p$-flowers is then $\ell m$.}\label{fig:new-gadgets}
\end{figure*}

\begin{lemma}\label{lmm:iff-circ-k}
For any $k$-maximum solution $\tau$ of graph $G_k(\mathcal{X},\mathcal{Y})$, it holds that:
(i) $\textup{s}_k(\tau) \leq |\mathcal{X}|+|\mathcal{Y}|+\|\mathcal{Y}\|$;
and (ii) formula $\mathcal{Y}$ has a satisfying assignment if and only if $\textup{s}_k(\tau) = |\mathcal{X}|+|\mathcal{Y}|+\|\mathcal{Y}\|$.
\end{lemma}

\begin{proof}
After the adaptation, the gadgets are plugged to $p$-flowers, where $p=\frac{k}{2}+1$. According to Lemma~\ref{lemma:p-flower}, any cycle that passes through a $p$-flower has a length of at least $k+2$. 
In $G_{k}(\mathcal{X},\mathcal{Y})$ the 8-cycles of $G_{8}(\mathcal{X},\mathcal{Y})$
are extended to $k$-cycles. More generally, note that all alternating cycles pass through flowers or through extended $\check{\mathbb{D}}$-edges, or both. This can be verified by examining each gadget and the connections between gadgets (see Fig.~\ref{fig:new-gadgets}):
\begin{itemize}
\item The outer connections of gadgets depicted in parts (a), (b), (d) and (e) always pass through a gadget of type $W$ (depicted in part (c)), and either go through a flower or through the extended $\check{\mathbb{D}}$-edge of~$W$. 
\item For the gadgets depicted in parts (a), (b) and (d), from the leftmost and the rightmost $\check{\mathbb{D}}$-edges, we can either go through a flower, or through an outer connection, or through one of the two middle squares. In the latter case, we can continue either through a flower, or through an outer connection, or through the extended $\check{\mathbb{D}}$-edge of the gadget.
\item For the gadget of a 3-clause depicted in part (e), besides the three cycles $\Theta_1$, $\Theta_2$ and $\Theta_3$ that contain one of the two extended $\check{\mathbb{D}}$-edges of the gadget, there are two cycles that have length 12 in $G_{8}(\mathcal{X},\mathcal{Y})$ (one of the two is depicted in Fig.~\ref{fig:12-cycles}\,(a)) and both of them include the middle extended $\check{\mathbb{D}}$-edge. Other possible alternating paths in this gadget either go through a flower, or through an outer connection, or through at least one of the two extended $\check{\mathbb{D}}$-edges. 
\end{itemize}

This shows that all cycles of $G_{8}(\mathcal{X},\mathcal{Y})$
have their lengths extended by at least $2\ell$ in
$G_{k}(\mathcal{X},\mathcal{Y})$, where $\ell=\frac{k-8}{2}$.
An 8-maximum solution $\tau'$ of $G_8(\mathcal{X},\mathcal{Y})$ with score $\textup{s}_8(\tau') = s$ gives a $k$-maximum solution $\tau$ of $G_k(\mathcal{X},\mathcal{Y})$ that, besides the long 
cycles with lengths at least $k+2$, has exactly the same $s$ 8-cycles induced by $\tau'$, but here each of these cycles has length $k$. Therefore, $\textup{s}_k(\tau) = s$ and $\tau$ has the same correspondence to the $\satvar$-SAT solution of $(\mathcal{X},\mathcal{Y})$.
\end{proof}




\begin{theorem}
Given a pair of $\cogsd$-cognate genomes $\mathbb{S}$ and $\mathbb{D}$, 
computing the $\sigma_k$ double distance or the $\sigma_k$ disambiguation of $\mathbb{S}$ and $\mathbb{D}$ is NP-complete for any finite even $k \geq 10$.
\end{theorem}

\begin{proof}
  The circular case follows directly from Lemma~\ref{lmm:iff-circ-k}. 
For the linear case, let $\nu_k$ be the number of vertices in $G_{k}(\mathcal{X},\mathcal{Y})$ and let $G^{+\nu}_{k}(\mathcal{X},\mathcal{Y})$ be the graph obtained by adding $\nu_k$ isolated vertices to $G_k(\mathcal{X},\mathcal{Y})$ (Construction~\ref{const:3satAmbLin}).
Note that a $k$-maximum solution $\tau$ of $G_k(\mathcal{X},\mathcal{Y})$, with $\textup{s}_k(\tau) = s$, gives a $k$-maximum solution $\tau''$ of $G^{+\nu}_k(\mathcal{X},\mathcal{Y})$ that, besides the same cycles of $\tau$, has $\nu_k$ 0-paths, that is $\textup{s}_k(\tau'') = s + \frac{\nu_k}{2}$.
Therefore, $\tau''$ also has the same correspondence to the $\satvar$-SAT solution of $(\mathcal{X},\mathcal{Y})$.
Since circular and linear genomes are special cases of mixed genomes, the result also holds for mixed genomes.
\end{proof}



\section{NP-completeness of \boldmath$\sigma_\infty$ disambiguation and of \boldmath$\sigma_\infty$ double distance}

Regarding (DCJ) $\sigma_\infty$ double distance, NP-hardness was proven by Tannier \textit{et al.}~\cite{tannier2009multichromosomal}, but they only considered circular genomes (that are a particular case and therefore cover mixed genomes as well).
First we remark that the problem is in NP since a solution for it can be easily verified in polynomial time because a certificate is a collection of cycles, of which one has to determine the cardinality, and for linear genomes also of paths, of which one has to count how many have even length.
Furthermore, the NP-completeness of $\sigma_\infty$ double distance (Theorem~4 in~\cite{tannier2009multichromosomal}) can be generalized to linear genomes, using the same strategy of the previous section. Therefore it can be completed as follows:

\begin{theorem}\label{thm:NP-DCJ-lin}
Given a pair of $\cogsd$-cognate linear genomes $\mathbb{S}$ and $\mathbb{D}$, computing the (DCJ) $\sigma_\infty$ double distance or the $\sigma_\infty$ disambiguation of $\mathbb{S}$ and $\mathbb{D}$ is NP-complete.
\end{theorem}
\begin{proof}
In the proof of NP-hardness for circular genomes, Tannier \textit{et al.} provided a reduction from bicolored graph decomposition by creating as an instance of $\sigma_\infty$ disambiguation a graph $H$ that, similarly to $G_8$ (given by Construction~\ref{const:3satAmb}), is composed of orange squares with black connections. This graph $H$ has $\nu=4q$ vertices, where $q$ is the number of vertices in the original bicolored graph, and can be easily modified into a graph $H^{+\nu}$ by adding $\nu$ isolated vertices to $H$ (analogously to Construction~\ref{const:3satAmbLin}).
Note that Proposition~\ref{prop:ABG-lin} also applies to $H^{+\nu}$, therefore it is easy to obtain a pair of $\cogsd$-cognate linear genomes $\mathbb{S}$ and $\mathbb{D}$ and a genome $\check{\mathbb{D}} \in \mathfrak{S}^\mathtt{a}_\mathtt{b}(\mathbb{D})$, such that  $H^{+\nu} = ABG(\mathbb{S},\check{\mathbb{D}})$.
The adaptation of the reduction given by Tannier \textit{et al.} is then similar to that we did for the bounded $k$: again the only difference is that the score of $\sigma_\infty$ disambiguation increases by $\frac{\nu}{2}$, but the relation to the instance of bicolored graph decomposition remains the same.
\end{proof}

\section{Summary and Conclusion}

The $\sigma_k$ distance between two canonical genomes, denoted by $\textup{d}_k$ for any even $k \geq 2$, is a family of measures that resides between the well known breakpoint ($\textup{d}_2$) and DCJ ($\textup{d}_\infty$) distances, and can be easily derived from the breakpoint graph of the input genomes.
This family is the basic unit for the $\sigma_k$ double distance, a family of optimization problems that minimize the $\sigma_k$ distance between the possible canonizations of pairs of duplicated genomes obtained from a $\cogsd$-cognate pair. 

Previously it was known that the $\sigma_k$ double distance could be solved in linear time for $k \in \{2,4,6\}$, and this result applies to circular, linear and mixed genomes~\cite{braga2024doubleDistance,tannier2009multichromosomal}.
It was also known that computing the $\sigma_\infty$ (DCJ) double distance was NP-hard for circular and mixed genomes~\cite{tannier2009multichromosomal}.
With the results presented here we closed the complexity gap of the $\sigma_k$ double distance of circular, linear and mixed genomes, by showing that it is NP-complete for any even finite $k \geq 8$ and for unbounded $k$.

An additional observation can be made concerning the related \emph{maximum independent set} (MIS) problem, that was the basis for the recently proposed linear time algorithms for the double distance under $\textup{d}_4$ and $\textup{d}_6$.
These algorithms consist of deriving from the ambiguous breakpoint graph an \emph{intersection graph} of the candidate cycles; and for $\textup{d}_4$ and $\textup{d}_6$, the obtained intersection graphs are particular linear-time solvable instances of MIS~\cite{braga2024doubleDistance}.
For any finite $k$, the intersection graph can be built in polynomial time. Since $\sigma_k$ disambiguation is NP-hard for $k\geq 8$, determining MIS on the intersection graph of the candidate cycles of length up to $k \geq 8$ is also NP-hard.

A wide area of research that is still open is determining how the picture looks for related problem families like the $\sigma_k$ median and the $\sigma_k$ guided halving problems that are also based on $\textup{d}_k$ and show a similar phase transition from polynomial ($\sigma_2$) to NP-hard ($\sigma_\infty$)~\cite{silva2023algorithms,tannier2009multichromosomal}.


\bibliographystyle{splncs04}
\bibliography{arxiv}

\vfill

\end{document}